\documentclass[a4paper,UKenglish,cleveref, autoref, thm-restate]{lipics-v2021}

\nolinenumbers 

\input xy 
\xyoption{all}
\usepackage{array}
\usepackage{adjustbox}
\usepackage{tikz-cd}
\usepackage{tikz}
\usetikzlibrary{cd,arrows}
\usepackage{mathpartir}
\usepackage{caption}
\usepackage{subcaption}
\usepackage{amsfonts}
\usepackage{amsmath}
\usepackage[renew-matrix,renew-dots]{nicematrix}
\usepackage{float}
\usepackage{enumitem}
\usepackage{xcolor}
\usepackage{booktabs}
\usepackage{rotating}
\usepackage{multirow}
\usepackage{comment}
\usepackage{graphicx}
\usepackage{braket}
\usepackage{enumitem}
\usepackage{xparse}
\usepackage{xstring}
\usepackage{mathtools}
\usepackage{xifthen}
\usepackage{makecell}
\usepackage{txfonts}
\usepackage{hyperref}
\usepackage{cleveref}
\hypersetup{
    colorlinks,
    linkcolor={black!50!black},
    citecolor={blue!50!black},
    urlcolor={blue!80!black}
}
\usepackage{mathtools}
\usepackage{changepage}

\usepackage{quiver}

\usepackage{todonotes}

\usepackage{tabularx}

\usepackage{wrapfig}
\input{macros}

\title{Completeness for Probabilistic Boolean Tapes}

\author{Filippo Bonchi}{University of Pisa}{}{https://orcid.org/0000-0002-3433-723X}{Bonchi is supported by the Ministero dell'Università e della Ricerca of Italy grant PRIN 2022 PNRR No. P2022HXNSC - RAP (Resource Awareness in Programming). This study was carried out within the National Centre on HPC, Big Data and Quantum Computing - SPOKE 10 (Quantum Computing) and received funding from the European Union Next-GenerationEU - National Recovery and Resilience Plan (NRRP) – MISSION 4 COMPONENT 2, INVESTMENT N. 1.4 – CUP N. I53C22000690001.}

\author{Cipriano Junior Cioffo}{University of Pisa}{}{https://orcid.org/0000-0002-4189-0930}{}

\authorrunning{F. Bonchi, C.J. Cioffo} 

\Copyright{Filippo Bonchi and Cipriano Junior Cioffo} 



\category{} 

\relatedversion{}

\funding{This research was partly funded by the Advanced Research + Invention Agency (ARIA) Safeguarded AI Programme.}

\EventEditors{John Q. Open and Joan R. Access}
\EventNoEds{2}
\EventLongTitle{}
\EventShortTitle{CVIT 2016}
\EventAcronym{CVIT}
\EventYear{2016}
\EventDate{December 24--27, 2016}
\EventLocation{Little Whinging, United Kingdom}
\EventLogo{}
\SeriesVolume{42}
\ArticleNo{23}

\begin{document}

\maketitle

\begin{abstract}
Probabilistic Boolean circuits have recently been proposed as a string-diagrammatic foundation for finite probabilistic programming. In this paper, we present a complete set of axioms for their semantics in terms of Markov kernels. 
Our approach is based on two intermediate results: completeness for \emph{partial} Boolean circuits and completeness for probabilistic Boolean tapes, a diagrammatic language for rig categories.
\end{abstract}

\section{Introduction}

Diagrammatic languages play a central role in computer science, with classical examples including Petri nets, graph rewriting systems~\cite{corradini1997algebraic,ferrari2005synchronised}, and data- and control-flow diagrams~\cite{bohm1966flow}, among others. 
More recently, the increasing importance of spatial structure in computation led Milner to move beyond the traditional term-based syntax of process calculi and introduce bigraphs~\cite{milner2009space}. Similar trends arise in quantum computing \cite{RaussendorfBriegel2001,danos2007measurement,Abramsky2008:CQM,coecke2011interacting,Backens2014,Hadzihasanovic2015} and probabilistic programming \cite{Kschischang2001,Darwiche2003,PoonDomingos2011,stein2024probabilistic,DBLP:journals/pacmpl/LiellCockS25}, where--unlike in classical computation--information cannot be freely copied or discarded. More broadly, the shift towards viewing data as a physical resource rather than a purely logical entity (see, e.g.,~\cite{orchard2019quantitative,gaboardi2016combining}) has driven the widespread adoption~\cite{BaezErbele-CategoriesInControl,DBLP:journals/pacmpl/BonchiHPSZ19,Bonchi2015,Fong2015,DBLP:journals/corr/abs-2009-06836,Ghica2016,DBLP:conf/lics/MuroyaCG18,Piedeleu2021,DBLP:journals/jacm/BonchiGKSZ22,cho2019disintegration,jacobs2019causal,piedeleu2025boolean,DBLP:journals/corr/abs-2410-10627,DBLP:journals/corr/abs-2502-03477,DBLP:journals/corr/abs-2301-12989,jacobs2021logical,moss2023category,fritz2021finetti,fritz2023dilations,fritz2018bimonoidal,stein2024probabilistic,perrone2023markov,fritz2023d,fritz2023weakly} of \emph{string diagrams}~\cite{joyal1991geometry,Selinger2009} as a graphical syntax.

Formally, string diagrams are morphisms in the symmetric monoidal category freely generated by a monoidal signature $\sign$, which we denote by $\Diag{\sign}$. The signature contains a basic set of generating symbols and complex diagrams are built via horizontal ($;$) and vertical ($\otimes$) composition. Such diagrams are typically interpreted in a monoidal category $\Cat{C}$ via a map $\osem{\cdot}\colon \Diag{\sign} \to \Cat{C}$ assigning semantics to each diagram. When $\osem{\cdot}$ is a monoidal functor--i.e., it preserves both $;$ and $\otimes$--the semantics is compositional, and the induced notion of semantic equality forms a congruence. Axiomatizing this congruence is a central challenge, as it enables equational reasoning directly at the level of diagrams.

In this paper, we provide a complete axiomatisation of \emph{probabilistic Boolean circuits}~\cite{piedeleu2025boolean}. These circuits are string diagrams generated by the following monoidal signature, where $p$ belongs to $(0,1)$.
\vspace{-0.28cm}
\begin{equation}\label{eq:fullsyntax}\tag{PrB}
\overbrace{
  \underbrace{\Andgate \quad \Notgate \quad \Flip{1} \quad \CBcopier \quad \CBdischarger}_{\tiny B} \quad \CBcocopier
}^{{\tiny PB}}
\quad
\Flip{p}
\end{equation}
\vspace{-0.09cm}
The first three generators are the standard Boolean gates: $\Andgate$ takes two inputs and returns their conjunction; $\Notgate$ maps an input to its negation; and $\Flip{1}$ produces the constant $1$. The generator $\CBcopier$ duplicates its input, while $\CBdischarger$ discards it. The gate $\CBcocopier$ can be seen as a dual of $\CBcopier$: it takes two inputs and returns their common value if they agree, and otherwise fails. Finally, $\Flip{p}$ is a probabilistic generator that outputs $1$ with probability $p$ and $0$ with probability $1-p$.
A circuit $c$ with $n$ inputs and $m$ outputs denotes a function $\osem{c}\colon 2^n \to \Dis(2^m)$, where $2 = \{0,1\}$ and $\Dis(X)$ is the set of subdistributions over $X$. Categorically, the semantics $\osem{\cdot}$ is a functor from $\Diag{PrB}$, the category of probabilistic circuits, to $\KlD$, the Kleisli category of the subdistribution monad (see e.g.~\cite{hasuo2007generic}).

The brace below~\eqref{eq:fullsyntax} highlights an important fragment: diagrams in $\Diag{B}$ correspond to standard Boolean circuits, hence denote Boolean functions. As such, they can be freely copied and discarded, meaning respectively that the following equalities hold:          
\[

    \InputIfFileExists{tapes/cipriano/natcopierpuntinileft.tikz}{}{\input{./tikz/tapes/cipriano/natcopierpuntinileft.tikz}}

=
\raisebox{0.1em}{
    \InputIfFileExists{tapes/cipriano/natcopierpuntiniright.tikz}{}{\input{./tikz/tapes/cipriano/natcopierpuntiniright.tikz}}
}
\qquad

    \InputIfFileExists{tapes/cipriano/natdiscardpuntinileft.tikz}{}{\input{./tikz/tapes/cipriano/natdiscardpuntinileft.tikz}}

=
\raisebox{0.1em}{
    \InputIfFileExists{tapes/cipriano/natdiscardpuntiniright.tikz}{}{\input{./tikz/tapes/cipriano/natdiscardpuntiniright.tikz}}
}
\]

The brace above~\eqref{eq:fullsyntax} identifies a larger fragment: diagrams in $\Diag{PB}$ denote Boolean \emph{partial} functions. Indeed, $\CBcocopier$ may fail to produce an output. In this setting, copying is still free, but discarding is not: for instance, $\CBcocopier ; \CBdischarger$ and $\CBdischarger \otimes \CBdischarger$ are not equal.
Crucially, general probabilistic Boolean circuits cannot even be copied. For example, duplicating the result of a coin toss, as in $\Flip{p} ; \CBcopier$, is not equivalent to tossing two independent coins, as in $\Flip{p} \otimes \Flip{p}$.

While we refer the reader to~\cite{piedeleu2025boolean} for a discussion of the practical relevance of this formalism, we emphasise here the advantages of diagrammatic syntax over traditional term-based representations. Standard Boolean identities fail in the probabilistic setting--for instance, $x \wedge x = x$ does not hold when $x \in \mathcal{D}(2)$. However, these laws hold in diagrammatic form, where copying and discarding are made explicit; see, e.g.,~\eqref{ax:B5} in \Cref{tab:booleanalgebra}. The latter collects the axioms of Boolean algebra which, as shown in~\cite{piedeleu2025boolean}, yields a complete axiomatisation of the Boolean fragment $\Diag{B}$.

\smallskip
 Our first contribution is a complete axiomatisation of partial Boolean circuits, i.e., circuits in $\Diag{PB}$. The axioms extend those of Boolean algebra in \Cref{tab:booleanalgebra} with additional laws governing $\CBcocopier$, presented in \Cref{tab:partialbooleanalgebra}. Completeness --in Theorem~\ref{thm:completenesspartialcircuits}-- is crucial for the second contribution. 

For the full language, rather than axiomatise $\Diag{PrB}$ directly, we consider a more expressive formalism--\emph{probabilistic Boolean tapes}-- which as illustrated in~\cite[Ex.~30]{bonchi2025tapediagramsmonoidalmonads} can faithfully express  probabilistic Boolean circuits. Probabilistic Boolean tapes arise from a general construction applied to partial Boolean circuits. Concretely, by freely adding \emph{convex biproducts} to $\Diag{PB}$, we obtain a category $\CatT{\Diag{\SigPB}}$ representing the syntax, together with a semantic map $\dsem{\cdot}\colon \CatT{\Diag{\SigPB}} \to \KlD$
extending $\osem{\cdot}\colon \Diag{\SigPB}\to \KlD$.
Crucially, $\CatT{\Diag{\SigPB}}$ is, like $\KlD$, a \emph{rig category}~\cite{laplaza_coherence_1972,johnson2021bimonoidal}, featuring two monoidal structures: $\otimes$ and $\oplus$. Moreover, $\dsem{\cdot}$ is a morphism of rig categories, i.e.\ it preserves $;$, $\otimes$ and $\oplus$. As a consequence, the semantics remains compositional.
Finally, probabilistic Boolean tapes admit an intuitive graphical syntax, closely related to string diagrams. The raison d'\^etre of tape diagrams~\cite{bonchi2023deconstructing} is precisely to represent the three forms of composition--$\otimes$, $\oplus$, and $;$--uniformly within two dimensions, rather than requiring an additional dimension as e.g. in~\cite{comfort2020sheet}.

The move from circuits to tapes allows us to reason not only about the monoidal product $\otimes$ in $\KlD$, but also about the operation $\oplus$, which plays a fundamental role in many probabilistic frameworks~\cite{stein2024probabilistic,introductioneffectus,DBLP:journals/pacmpl/LiellCockS25}. As discussed in Example~\ref{ex:esempioSTRONG}, the presence of $\oplus$ provides a natural way of modelling probabilistic control and, perhaps more importantly, leads to a structurally transparent axiomatisation. In addition to the equalities induced by the tape construction in~\Cref{fig:tapesax}, our axiomatisation comprises the three additional laws shown in Figure~\ref{ax:BooleanTAPES}. The third of these is an implication expressing \emph{cancellativity}. Although not purely equational, cancellativity is a standard reasoning principle that arises naturally in a wide variety of algebraic structures.


Relying on completeness for partial Boolean circuits and on a result from~\cite{probbooltapesfossacs}--stated here as \Cref{cor:isotapematrices}--characterising tape diagrams as stochastic matrices of string diagrams, we prove completeness for probabilistic Boolean tapes in \Cref{cor:completenessPBPtapes}. Consequently, as stated in \Cref{cor:finale}, semantic equality of probabilistic Boolean circuits reduces to equality of their encodings as tapes. 

\smallskip
\noindent{\bf Related Work.} Figures~4 and~5 of~\cite{piedeleu2025boolean} provide an axiomatisation of probabilistic Boolean circuits. Crucially, the induced equivalence is strictly coarser than semantic equality, i.e., the one induced by $\osem{-}$. In fact, two diagrams $c$ and $d$ are equivalent in~\cite{piedeleu2025boolean} if and only if $\osem{c} \propto \osem{d}$, where $\propto$ is the relation on $\KlD$ defined by $f \propto g$ whenever there exists $\lambda > 0$ such that
$f(x)(y) = \lambda \cdot g(x)(y)$ for all $x \in X$ and $y \in Y$.
Hence, several axioms in~\cite{piedeleu2025boolean}--in particular F2 and F7--are unsound with respect to $\osem{\cdot}$. Interestingly, as shown in \Cref{lemma:failureequalities}, axiom F8 of~\cite{piedeleu2025boolean} is derivable from our system.

To the best of our knowledge, partial Boolean circuits have not previously been axiomatised. In contrast,~\cite{piedeleu2025boolean} gives a complete axiomatisation for the \emph{causal} fragment, i.e., circuits generated by all gates in~\eqref{eq:fullsyntax} except $\CBcocopier$. In~\cite{probbooltapesfossacs}, the characterisation of tapes as stochastic matrices is used to derive an alternative axiomatisation for the same fragment.
Despite exploiting the same result, the two axiomatisations differ substantially. Indeed, among all the axioms in Figures~\ref{tab:partialbooleanalgebra} and~\ref{ax:BooleanTAPES}, (T3) is the only one that is even expressible with the tapes in~\cite{probbooltapesfossacs}.

\smallskip
\noindent {\bf Synopsis. } Our presentation of probabilistic Boolean circuits is staged: we begin in \Cref{sec:boolcircuits} with $\Diag{\SigB}$, then $\Diag{\SigPB}$ in 
\Cref{ssec:partialboolean} and finally $\Diag{\SigPRB}$  in \Cref{ssec:probabilisticboolean}. Hence, the axiomatisation for partial Boolean circuits and its proof of completeness are illustrated in  \Cref{ssec:partialboolean}.  \Cref{sec:encodingprobboolcircuits} recalls from \cite{bonchi2025tapediagramsmonoidalmonads} probabilistic Boolean tapes as well as the encoding of circuits into tapes; \Cref{sec:preliminaries} illustrates several general results from \cite{probbooltapesfossacs} about the tape construction. These are used in  \Cref{sec:completenessprobbooltapes} to provide a complete axiomatisation of probabilistic Boolean tapes and prove its completeness. All missing proofs are in the appendices. Appendix \ref{app:sec:preliminaries} contains additional material: in particular the axioms in Table \ref{tabellone} already appear in the form of tapes in \Cref{fig:tapesax} in the main text.

\section{Boolean Circuits}\label{sec:boolcircuits}
We commence our exposition by recalling how Boolean circuits can be regarded as string diagrams.

A \emph{monoidal signature} is a tuple $(\sort, \sign, \ari, \coar)$ where $\sort$ is a set of basic sorts, hereafter denoted by $A,B,\dots$, $\sign$ is a set of generators, denoted by $s,t, \dots$, and $\ari,\coar \colon \sign \to \sort^*$ assign to each symbol its arity and coarity, i.e., words over $\sort$, denoted by $U,V, \dots$. We consider terms generated by the following context free grammar, where $A,B \in \sort$ and $s \in \Sigma$:
\begin{equation*}\label{stringdiagramGrammar}
\setlength{\arraycolsep}{3pt} 
\begin{array}{rcccccccccccccccc}
c & ::= & \id{A} & \mid & \id{\uno} & \mid & \gen & \mid & \symmt{A}{B} & \mid & c ; c & \mid & c \per c
\end{array}
\end{equation*}
The rules in Table~\ref{fig:freestricmmoncatax} assigns a type $U\to V$ to terms. Well-typed terms, taken modulo the axioms in the same table, are the arrows of $(\DiagS, \per, \uno)$, the strict symmetric monoidal category freely generated by $\sign$. Its objects are words in $\sort^*$. Arrows of $\DiagS$ admit a graphical representation in terms of \emph{string diagrams} \cite{joyal1991geometry,selinger2010survey}. The grammar above can be depicted diagrammatically as follows:
\begin{equation*}
    \setlength{\tabcolsep}{2pt}
    \begin{tabular}{rc c@{$\,\mid\,$}c@{$\,\mid\,$}c@{$\,\mid\,$}c@{$\,\mid\,$}c@{$\,\mid\,$}c@{$\,\mid\,$}c@{$\,\mid\,$}c@{$\,\mid\,$}c}
        $c$  & $\Coloneqq$ &  $\wire{A}$ & $ 
    \InputIfFileExists{empty.tikz}{}{\input{./tikz/empty.tikz}}
 $ & $ \Cgen{\gen}{A}{B} $ & $ \Csymm{A}{B} $ & $ 
    \InputIfFileExists{seq_compC.tikz}{}{\input{./tikz/seq_compC.tikz}}
 $ & $ 
    \InputIfFileExists{par_compC.tikz}{}{\input{./tikz/par_compC.tikz}}
$ \\
    \end{tabular}
\end{equation*}
The diagrammatic notation internalises the axioms in Table~\ref{fig:freestricmmoncatax}, yielding proofs that are both more concise and more intuitive. Throughout the paper, we will also rely on the term-based notation, which is often more succinct for inductive definitions.

A \emph{monoidal theory} $\mathbb{T}=(\sign, E)$ consists of a monoidal signature $\sign$ and a set $E$ of pairs of arrows of $\DiagS$ of the same type. Let $=_{\mathbb{T}}$ be the congruence closure (w.r.t. $;$ and $\otimes$) of $E$. We write $\DiagT$ for the monoidal category obtained as the quotient of $\DiagS$ by $=_{\mathbb{T}}$ and $Q_{\mathbb{T}}\colon \Diag{ \sign} \to \Diag{\mathbb{T}}$ for the monoidal functor mapping each diagram into its $=_{\mathbb{T}}$ equivalence class.

\begin{table}[H]
\resizebox{\textwidth}{!}{%
$
\begin{array}{c}
\toprule
\begin{array}{c}
    \begin{array}{ccc}
     {id_\uno \colon \uno \to \uno} &
        {id_A \colon A \to A} &
        {\sigma_{A, B}^{\per} \colon A \per B \to B \per A} 
   \\[0.3em]
        \inferrule{\ari(\gen)=U \quad \coar(\gen)= V}{\gen \colon U \to V} &
        \inferrule{c \colon U \to V \and d \colon V \to W}{c ; d \colon U \to W} &
        \inferrule{c \colon U_1 \to V_1 \and d \colon U_2 \to V_2}{c \per d\colon U_1 \per U_2 \to V_1 \per V_2}       
    \end{array}
\end{array}
\\[0.5em]
\begin{array}{cc}
\midrule
\begin{array}{c}
(c;d);e=c;(d;e) \qquad id_U;c=c=c;id_Q\\
(c_1\per c_2) ; (d_1 \per d_2) = (c_1;d_1) \per (c_2;d_2)
\end{array} 
&
\begin{array}{c}
id_{\uno}\per c = c = c \per id_{\uno} \qquad (c \per d)\, \per e = c \per \,(d \per e) \\
\sigma_{U, V}; \sigma_{V, U}= id_{U \per V} \qquad (\gen \per id_W) ; \sigma_{V, W} = \sigma_{U,W} ; (id_W \per \gen)
\end{array}
\end{array}\\
\bottomrule
\end{array}
$
}
\caption{Typing rules (top) and axioms (bottom) for freely generated strict symmetric monoidal categories.}
\label{fig:freestricmmoncatax}
\end{table}

Consider the monoidal signature with a single sort, $\sort \defeq \{ A\}$, and  generators 
\[ \SigB \defeq \{\Andgate  , \Notgate , \Flip{1},   \CBcopier, \CBdischarger \} \text{.}\]  
Arities and coarities are determined by the number of ports on the left and on the right: for instance, $\Andgate$ has arity $A^2=AA$ and coarity $A$, while $\CBdischarger $ has arity $A$ and coarity $A^0 =1$. 

The first three generators of $\SigB$ correspond to the operations and constants of Boolean algebras, namely $\wedge$, $\neg$, and $1$. 
The generator \emph{copy}, denoted by $\CBcopier$, takes a Boolean signal as input and produces two identical outputs. 
The generator \emph{discard}, denoted by $\CBdischarger$, takes a Boolean signal as input and discards it.
To make this  formal, we interpret each generator $s\in \SigB$ with arity $A^n$ and coarity $A^m$ as function $\osemB{s}\colon 2^n \to 2^m$ where $2$ is the set of Booleans  $\{0,1\}$ and $2^0$ is the singleton set $1\defeq \{\bullet\}$.
\begin{equation}\label{eq:sembool}
    \begin{array}{rclrclrcl}
      \osemB{\Andgate}\colon 2 \times 2& \to & 2  & \osemB{\Notgate}\colon  2& \to & 2 & \osemB{\Flip{1}} \colon 1 & \to & 2 \\
      (x,y)&\mapsto& x\wedge y &  x & \mapsto & \neg x & \bullet &\mapsto&1\\[4pt]
       \osemB{\CBcopier} \colon 2  & \to & 2\times 2 & \osemB{\CBdischarger} \colon 2 & \to & 1 & 
       \\
      x &\mapsto&(x,x) & x &\mapsto&{\bullet} & 
    \end{array}
  \end{equation}
Such interpretation gives rise to a symmetric monoidal functor $\osemB{-}$ from $(\Diag{\SigB}, \per, \uno)$ to $(\Sets,\per, \uno)$. The latter is the category of sets and functions where the monoidal tensor is simply cartesian product.
The  functor  $\osemB{-} \colon (\Diag{\SigB}, \per, \uno) \to (\Sets,\per, \uno)$ is defined on objects as $\osemB{A^n} \defeq 2^n$ and on arrows as
\begin{equation}\label{eq:smfunctor}
\osemB{\id{A}}\defeq \id{2} \;\;   \osemB{\id{\uno}}\defeq \id{1} \;\;  \osemB{\symmt{A}{A}}\defeq \symmt{2}{2}
 \;\; \osemB{ c ; d }\defeq \osemB{c} ; \osemB{d} \;\; \osemB{c \per d} = \osemB{c} \per \osemB{d}
%
\end{equation}

We regard diagrams as syntax and $\osemB{-}$ as the map assigning semantics to syntax. Semantic equality--i.e., the congruence induced by $\osemB{-}$-- is axiomatised by the laws in Figure~\ref{tab:booleanalgebra}.
The equalities in the first block are the standard axioms of Boolean algebras, where 
the constant $\Flip{0}$ and the \emph{or} gate $\Orgate$ are defined as expected (see Figure~\ref{table:multiplexer}). 
Note that, in the string-diagrammatic setting, the structural maps $\CBcopier$ and $\CBdischarger$ must be made explicit, for instance in \eqref{ax:B6} that is the usual law of non contradiction.
The equalities in the second block state that $\CBcopier$ and $\CBdischarger$ form a commutative comonoid. 
The equalities in the third and fourth blocks impose naturality of $\CBcopier$ and $\CBdischarger$, respectively.

We write $\ThB$ for the monoidal theory consisting of the signature $\SigB$ and the axioms  in Figure~\ref{tab:booleanalgebra}.


\begin{theorem}[From \cite{piedeleu2025boolean}]\label{thm:completenessbooleancircuits}
For all  $c,d\in \Diag{\SigB}$, $\osemB{c}=\osemB{d}$ iff $c\eqB d$.
\end{theorem}

\begin{figure}
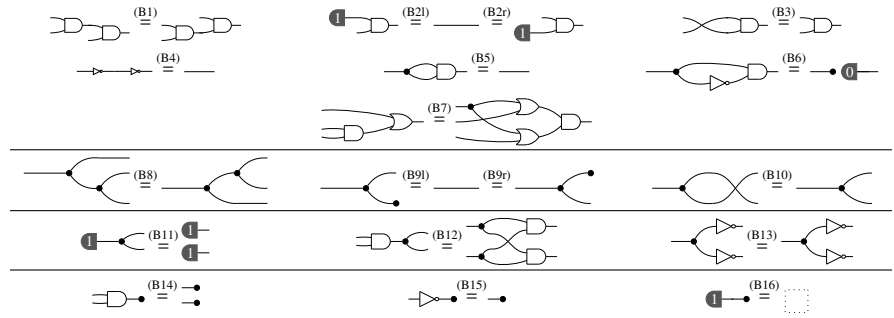

  \scalebox{0.8}{%
  \begin{tabular}{ccccc}


  \mylabelbis{ax:B1}{B1}%
  $
    \InputIfFileExists{tapes/cipriano/axB1left.tikz}{}{\input{./tikz/tapes/cipriano/axB1left.tikz}}
\raisebox{-1.5ex}{\,$\stackrel{\eqref{ax:B1}}{=}$\,}
    \InputIfFileExists{tapes/cipriano/axB1right.tikz}{}{\input{./tikz/tapes/cipriano/axB1right.tikz}}
$
  & &
  \mylabelbis{ax:B2l}{B2l}\mylabelbis{ax:B2r}{B2r}%
  $
    \InputIfFileExists{tapes/cipriano/axB2left.tikz}{}{\input{./tikz/tapes/cipriano/axB2left.tikz}}
\raisebox{-1.5ex}{\,$\stackrel{\eqref{ax:B2l}}{=}$\,}
    \InputIfFileExists{tapes/cipriano/axB2middle.tikz}{}{\input{./tikz/tapes/cipriano/axB2middle.tikz}}
\raisebox{-1.5ex}{\,$\stackrel{\eqref{ax:B2r}}{=}$\,}
    \InputIfFileExists{tapes/cipriano/axB2right.tikz}{}{\input{./tikz/tapes/cipriano/axB2right.tikz}}
$
  & &
  \mylabelbis{ax:B3}{B3}%
  $
    \InputIfFileExists{tapes/cipriano/axB3left.tikz}{}{\input{./tikz/tapes/cipriano/axB3left.tikz}}
\raisebox{-1.5ex}{\,$\stackrel{\eqref{ax:B3}}{=}$\,}
    \InputIfFileExists{tapes/cipriano/axB3right.tikz}{}{\input{./tikz/tapes/cipriano/axB3right.tikz}}
$
  \\

  \mylabelbis{ax:B4}{B4}%
  $
    \InputIfFileExists{tapes/cipriano/axB4left.tikz}{}{\input{./tikz/tapes/cipriano/axB4left.tikz}}
\raisebox{-1.5ex}{\,$\stackrel{\eqref{ax:B4}}{=}$\,}
    \InputIfFileExists{tapes/cipriano/axB4right.tikz}{}{\input{./tikz/tapes/cipriano/axB4right.tikz}}
$
  & &
  \mylabelbis{ax:B5}{B5}%
  $
    \InputIfFileExists{tapes/cipriano/axB5left.tikz}{}{\input{./tikz/tapes/cipriano/axB5left.tikz}}
\raisebox{-1.5ex}{\,$\stackrel{\eqref{ax:B5}}{=}$\,}
    \InputIfFileExists{tapes/cipriano/axB5right.tikz}{}{\input{./tikz/tapes/cipriano/axB5right.tikz}}
$
  & &
  \mylabelbis{ax:B6}{B6}%
  $
    \InputIfFileExists{tapes/cipriano/axB6left.tikz}{}{\input{./tikz/tapes/cipriano/axB6left.tikz}}
\raisebox{-1.5ex}{\,$\stackrel{\eqref{ax:B6}}{=}$\,}
    \InputIfFileExists{tapes/cipriano/axB6right.tikz}{}{\input{./tikz/tapes/cipriano/axB6right.tikz}}
$
  \\

  & &
  \mylabelbis{ax:B7}{B7}%
  $
    \InputIfFileExists{tapes/cipriano/axB7left.tikz}{}{\input{./tikz/tapes/cipriano/axB7left.tikz}}
\raisebox{-1.5ex}{\,$\stackrel{\eqref{ax:B7}}{=}$\,}
    \InputIfFileExists{tapes/cipriano/axB7right.tikz}{}{\input{./tikz/tapes/cipriano/axB7right.tikz}}
$
  & &
  \\

  \midrule


  \mylabelbis{ax:A1}{B8}%
  $
    \InputIfFileExists{tapes/cipriano/axA1left.tikz}{}{\input{./tikz/tapes/cipriano/axA1left.tikz}}
\raisebox{-1.5ex}{\,$\stackrel{\eqref{ax:A1}}{=}$\,}
    \InputIfFileExists{tapes/cipriano/axA1right.tikz}{}{\input{./tikz/tapes/cipriano/axA1right.tikz}}
$
  & &
  \mylabelbis{ax:A2l}{B9l}\mylabelbis{ax:A2r}{B9r}%
  $
    \InputIfFileExists{tapes/cipriano/axA2left.tikz}{}{\input{./tikz/tapes/cipriano/axA2left.tikz}}
\raisebox{-1.5ex}{\,$\stackrel{\eqref{ax:A2l}}{=}$\,}
    \InputIfFileExists{tapes/cipriano/axA2middle.tikz}{}{\input{./tikz/tapes/cipriano/axA2middle.tikz}}
\raisebox{-1.5ex}{\,$\stackrel{\eqref{ax:A2r}}{=}$\,}
    \InputIfFileExists{tapes/cipriano/axA2right.tikz}{}{\input{./tikz/tapes/cipriano/axA2right.tikz}}
$
  & &
  \mylabelbis{ax:A3}{B10}%
  $
    \InputIfFileExists{tapes/cipriano/axA3left.tikz}{}{\input{./tikz/tapes/cipriano/axA3left.tikz}}
\raisebox{-1.5ex}{\,$\stackrel{\eqref{ax:A3}}{=}$\,}
    \InputIfFileExists{tapes/cipriano/axA3right.tikz}{}{\input{./tikz/tapes/cipriano/axA3right.tikz}}
$
  \\

  \midrule


  \mylabelbis{ax:C1}{B11}%
  $
    \InputIfFileExists{tapes/cipriano/axC1left.tikz}{}{\input{./tikz/tapes/cipriano/axC1left.tikz}}
\raisebox{-0.6ex}{\,$\stackrel{\eqref{ax:C1}}{=}$\,}
    \InputIfFileExists{tapes/cipriano/axC1right.tikz}{}{\input{./tikz/tapes/cipriano/axC1right.tikz}}
$
  & &
  \mylabelbis{ax:C2}{B12}%
  $
    \InputIfFileExists{tapes/cipriano/axC2left.tikz}{}{\input{./tikz/tapes/cipriano/axC2left.tikz}}
\raisebox{-0.6ex}{\,$\stackrel{\eqref{ax:C2}}{=}$\,}
    \InputIfFileExists{tapes/cipriano/axC2right.tikz}{}{\input{./tikz/tapes/cipriano/axC2right.tikz}}
$
  & &
  \mylabelbis{ax:C3}{B13}%
  $
    \InputIfFileExists{tapes/cipriano/axC3left.tikz}{}{\input{./tikz/tapes/cipriano/axC3left.tikz}}
\raisebox{-0.6ex}{\,$\stackrel{\eqref{ax:C3}}{=}$\,}
    \InputIfFileExists{tapes/cipriano/axC3right.tikz}{}{\input{./tikz/tapes/cipriano/axC3right.tikz}}
$
  \\

  \midrule


  \mylabelbis{ax:D1}{B14}%
  $
    \InputIfFileExists{tapes/cipriano/axD1left.tikz}{}{\input{./tikz/tapes/cipriano/axD1left.tikz}}
\raisebox{-0.6ex}{\,$\stackrel{\eqref{ax:D1}}{=}$\,}
    \InputIfFileExists{tapes/cipriano/axD1right.tikz}{}{\input{./tikz/tapes/cipriano/axD1right.tikz}}
$
  & &
  \mylabelbis{ax:D2}{B15}%
  $
    \InputIfFileExists{tapes/cipriano/axD2left.tikz}{}{\input{./tikz/tapes/cipriano/axD2left.tikz}}
\raisebox{-0.6ex}{$\stackrel{\eqref{ax:D2}}{=}$}
    \InputIfFileExists{tapes/cipriano/axD2right.tikz}{}{\input{./tikz/tapes/cipriano/axD2right.tikz}}
$
  & &
  \mylabelbis{ax:D3}{B16}%
  $
    \InputIfFileExists{tapes/cipriano/axD3left.tikz}{}{\input{./tikz/tapes/cipriano/axD3left.tikz}}
\raisebox{-0.6ex}{$\stackrel{\eqref{ax:D3}}{=}$}
  \raisebox{-0.5em}{
    \InputIfFileExists{empty.tikz}{}{\input{./tikz/empty.tikz}}
}$
  \\

  \end{tabular}}%
  \caption{The monoidal theory of Boolean algebras.}
  \label{tab:booleanalgebra}
\end{figure}

We conclude this section by fixing some syntactic sugar, collected in Figure~\ref{table:multiplexer}. 
The \emph{xnor} gate $\xnor$  takes two inputs and outputs $1$ if they are equal and $0$ otherwise. 
The \emph{multiplexer} $
    \InputIfFileExists{tapes/cipriano/multi.tikz}{}{\input{./tikz/tapes/cipriano/multi.tikz}}
$ takes three inputs and returns one output: if the first input is $1$, then it outputs the second input, otherwise the third input. 
The $m$-ary multiplexer $\Ifgatem$ works similarly, but it takes $2m+1$ inputs: if the first input is $1$, then it outputs the first $m$ inputs, otherwise it outputs the second $m$ inputs. 
Figure~\ref{table:multiplexer} also defines $m$-ary discard $\CBdischargerm$ and copier $\mcopier$. Note that, for the sake of readability of string diagrams, we label wires by  some natural number $n$ as an abbreviation for $A^n$.

\begin{lemma}\label{lemma:multiplexer}
 Let $c$ be a diagram in $\Diag{\SigB}[n,m]$. The following derived laws hold:
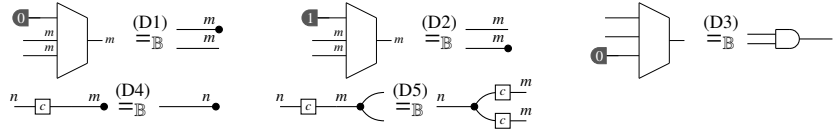
\mylabelbis{ax:M1}{D1}
\mylabelbis{ax:M2}{D2}
\mylabelbis{ax:M3}{D3}
\mylabelbis{ax:N1}{D4}
\mylabelbis{ax:N2}{D5}
{\setlength{\abovedisplayskip}{0.3em}
 \setlength{\belowdisplayskip}{0.3em}
\[
\begin{array}{c@{\qquad}c@{\qquad}c}
\raisebox{-0.3em}{\scalebox{0.8}{
    \InputIfFileExists{tapes/cipriano/D1left.tikz}{}{\input{./tikz/tapes/cipriano/D1left.tikz}}
}} 
\stackrel{\eqref{ax:M1}}{\eqB}

    \InputIfFileExists{tapes/cipriano/D1right.tikz}{}{\input{./tikz/tapes/cipriano/D1right.tikz}}

&
\raisebox{-0.3em}{\scalebox{0.8}{
    \InputIfFileExists{tapes/cipriano/D2left.tikz}{}{\input{./tikz/tapes/cipriano/D2left.tikz}}
}}
\stackrel{\eqref{ax:M2}}{\eqB}

    \InputIfFileExists{tapes/cipriano/D2right.tikz}{}{\input{./tikz/tapes/cipriano/D2right.tikz}}

&
\raisebox{-0.3em}{\scalebox{0.8}{
    \InputIfFileExists{tapes/cipriano/D3left.tikz}{}{\input{./tikz/tapes/cipriano/D3left.tikz}}
}}
\stackrel{\eqref{ax:M3}}{\eqB}

    \InputIfFileExists{tapes/cipriano/D3right.tikz}{}{\input{./tikz/tapes/cipriano/D3right.tikz}}

\\[0.5em]

    \InputIfFileExists{tapes/cipriano/D4left.tikz}{}{\input{./tikz/tapes/cipriano/D4left.tikz}}

\stackrel{\eqref{ax:N1}}{\eqB}

    \begin{tikzpicture}
	\begin{pgfonlayer}{nodelayer}
		\node [style=none] (11) at (0, 0) {};
		\node [style=none] (12) at (3, 0) {};
		\node [style=black] (13) at (3, 0) {};
		\node [style=none, scale=0.6] (19) at (2.5, 0.5) {$n$};
	\end{pgfonlayer}
	\begin{pgfonlayer}{edgelayer}
		\draw (11.center) to (12.center);
	\end{pgfonlayer}
\end{tikzpicture}
}

&

    \InputIfFileExists{tapes/cipriano/D5left.tikz}{}{\input{./tikz/tapes/cipriano/D5left.tikz}}

\stackrel{\eqref{ax:N2}}{\eqB}

    \InputIfFileExists{tapes/cipriano/D5right.tikz}{}{\input{./tikz/tapes/cipriano/D5right.tikz}}

&
\end{array}
\]}
\end{lemma}

\begin{figure}[t]
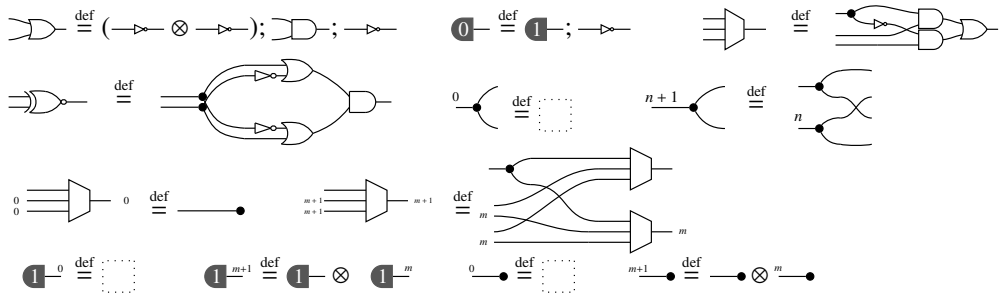

\centering
\resizebox{0.96\linewidth}{!}{%
\begin{tabular}{@{}l@{}}

$\Orgate \defeq (\Notgate \otimes \Notgate); \Andgate; \Notgate 
\qquad 
\Flip{0} \defeq \Flip{1}; \Notgate
\qquad

    \InputIfFileExists{tapes/cipriano/multi.tikz}{}{\input{./tikz/tapes/cipriano/multi.tikz}}
 
\defeq 

    \InputIfFileExists{tapes/cipriano/multiexpl.tikz}{}{\input{./tikz/tapes/cipriano/multiexpl.tikz}}
$ \\[4pt]

$
    \InputIfFileExists{tapes/cipriano/xnorgate.tikz}{}{\input{./tikz/tapes/cipriano/xnorgate.tikz}}

\qquad
\raisebox{-1em}{
    \InputIfFileExists{tapes/cipriano/ncopierzero.tikz}{}{\input{./tikz/tapes/cipriano/ncopierzero.tikz}}
} 
\raisebox{-0.5em}{$\defeq$}
\raisebox{-0.5em}{
    \InputIfFileExists{empty.tikz}{}{\input{./tikz/empty.tikz}}
}
\qquad 
\raisebox{-1em}{
    \InputIfFileExists{tapes/cipriano/ncopier.tikz}{}{\input{./tikz/tapes/cipriano/ncopier.tikz}}
}$ \\[6pt]

$\raisebox{-1.2em}{
    \InputIfFileExists{tapes/cipriano/mmultiplexerleftzero.tikz}{}{\input{./tikz/tapes/cipriano/mmultiplexerleftzero.tikz}}
} 
\defeq 

    \InputIfFileExists{tapes/cipriano/mmultiplexerrightzero.tikz}{}{\input{./tikz/tapes/cipriano/mmultiplexerrightzero.tikz}}

\qquad 
\raisebox{-1.2em}{
    \InputIfFileExists{tapes/cipriano/mmultiplexerleft.tikz}{}{\input{./tikz/tapes/cipriano/mmultiplexerleft.tikz}}
} 
\defeq 

    \InputIfFileExists{tapes/cipriano/mmultiplexerright.tikz}{}{\input{./tikz/tapes/cipriano/mmultiplexerright.tikz}}
$ \\[6pt]

$\mFlip{1}{0} \defeq 
    \InputIfFileExists{empty.tikz}{}{\input{./tikz/empty.tikz}}

\qquad
\mFlip{1}{m+1} \defeq \Flip{1} \otimes \mFlip{1}{m}
\qquad
\mCBdischarger{0} \defeq 
    \InputIfFileExists{empty.tikz}{}{\input{./tikz/empty.tikz}}

\qquad
\mCBdischarger{m+1} \defeq \CBdischarger\otimes \mCBdischarger{m}$

\end{tabular}
}
\caption{Definitions of or gate, $\Flip{0}$ and multiplexer (top); definitions of xnor gate and inductive definition of $n$-ary copier (middle);  $m$-ary multiplexer (third row); $\Flipm{1}$ and $n$-ary dischargers (bottom).}
\label{table:multiplexer}
\end{figure}

%
%
%
%
\section{Partial Boolean Circuits}\label{ssec:partialboolean}
Now consider the monoidal signature $\SigPB\defeq \SigB \cup\{\CBcocopier\}$ obtained by adding to $\SigB$ the generator \emph{cocopy} $\CBcocopier$.
Intuitively, $\CBcocopier$ is the dual of $\CBcopier$: it compares two inputs and, if these are equal, it outputs that value; otherwise, it produces no output. Thus, $\CBcocopier$ does not denote a function but a \emph{partial function}.

Let $(\Cat{Par}, \otimes, 1)$  be the category of sets and partial functions, where $\otimes$  is the cartesian product of sets and, for all partial functions $f_1\colon X_1 \to Y_1$ and $f_2\colon X_2 \to Y_2$, $f_1\otimes f_2 \colon X_1 \times X_2 \to Y_1 \times Y_2$ is defined for all $(x_1,x_2)\in X_1\times X_2$ as
\begin{equation}\label{eq:productpar}
f_1\otimes f_2 (x_1,x_2) \defeq \begin{cases} (\,f_1(x_1), \, f_2(x_2)\,) & \text{ if }f_1(x_1)\neq \bot\text{, } f_2(x_2)\neq \bot \\
 \bot & \text{ otherwise}\end{cases}
\end{equation}
Above and in the rest of the paper, we write $f(x)=\bot$ to say that $f$ is undefined on $x$. 
The obvious injection $\Pa\colon (\Sets, \per, \uno) \to (\Cat{Par}, \per, \uno)$ is easily proved to be a symmetric monoidal functor.

The generators in $\SigPB$ can all be interpreted as arrows of $\Cat{Par}$: $\CBcocopier $ is interpreted as
  \begin{equation}\label{eq:semcocopier}
    \begin{array}{rcl}
      \osemPB{\CBcocopier}\colon 2 \times 2& \to & 2   \\
      (x,y)&\mapsto& \begin{cases} x &\text{if } x=y\\ \bot & \text{else}\end{cases}    \end{array}
  \end{equation}
while generators $s \in \SigB$ as $\osemPB{s}\defeq \Pa(\osemB{s})$. The inductive extension, analogous to \eqref{eq:smfunctor}, of the above interpretation provides a symmetric monoidal functor $\osemPB{-}\colon \Diag{\SigPB} \to \Cat{Par}$ mapping each diagram in its semantics. For example, the semantics of the following diagrams
\begin{equation}\label{def:bottom}
  \raisebox{0cm}{
    \InputIfFileExists{tapes/cipriano/defcoflip.tikz}{}{\input{./tikz/tapes/cipriano/defcoflip.tikz}}
}
\end{equation}
is illustrated below.
  \[
    \begin{array}{rclrclrcl}
      \osemPB{\coflip{1} }\colon 2 & \to & 1  & \osemPB{\coflip{0} }\colon 2 & \to & 1 & \osemPB{\Flip{\bot}}\colon 1 & \to & 2 \\
      x&\mapsto& \begin{cases} \bullet &\text{if } x=1\\ \bot & \text{else}\end{cases}    
      &
      x&\mapsto& \begin{cases} \bullet &\text{if } x=0\\ \bot & \text{else}\end{cases}    
      &
      \bullet &\mapsto & \bot     
      \end{array}
  \]

In this section, we present  a complete axiomatization of the equality induced by $\osemPB{-}$.
The axioms are illustrated in Figure~\ref{tab:partialbooleanalgebra}: the equalities in the first row assert that $\CBcocopier$ is associative, commutative and idempotent; in the second row, \eqref{ax:F4} and \eqref{ax:F5} are the usual Frobenius equalities ruling the interaction of $\CBcocopier$ with $\CBcopier$; the remaining axioms  are more peculiar: \eqref{ax:F7} states that $x\wedge y = 1$ iff $x=y=1$; \eqref{ax:F6} provides a decomposition of  $\CBcocopier$ in terms of $\Andgate$, 
    \InputIfFileExists{tapes/cipriano/xnorgatealone.tikz}{}{\input{./tikz/tapes/cipriano/xnorgatealone.tikz}}
 and $\coflip{1}$. 
We write $\ThPB$ for the monoidal theory consisting of the signature $PB$ and the axioms in Figures~\ref{tab:booleanalgebra} and~\ref{tab:partialbooleanalgebra}. Simple computations confirm that the axioms are sound.
\begin{proposition}[Soundness]\label{prop:soundnesspartialb}
For all $c,d\in \Diag{PB}$, if $c\eqPB d$, then $\osemPB{c}=\osemPB{d}$.
\end{proposition}

\begin{figure}
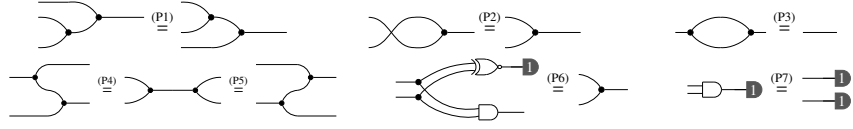

\centering
\scalebox{0.8}{
{
\setlength{\tabcolsep}{12pt}      
\setlength{\extrarowheight}{9pt}  
\renewcommand{\arraystretch}{1.0}

\begin{tabular}{ccc}

\hspace{1.5em}%
\mylabelbis{ax:F1}{P1}%
$
    \InputIfFileExists{tapes/cipriano/axF1left.tikz}{}{\input{./tikz/tapes/cipriano/axF1left.tikz}}

 \raisebox{-0.6em}{\,\,$\stackrel{\eqref{ax:F1}}{=}$\,\,}
 
    \InputIfFileExists{tapes/cipriano/axF1right.tikz}{}{\input{./tikz/tapes/cipriano/axF1right.tikz}}
$
&
\hspace{1.5em}%
\mylabelbis{ax:F2}{P2}%
$
    \InputIfFileExists{tapes/cipriano/axF2left.tikz}{}{\input{./tikz/tapes/cipriano/axF2left.tikz}}

 \raisebox{-0.6em}{\,$\stackrel{\eqref{ax:F2}}{=}$\,}
 
    \InputIfFileExists{tapes/cipriano/axF2right.tikz}{}{\input{./tikz/tapes/cipriano/axF2right.tikz}}
$
&
\hspace{1.5em}%
\mylabelbis{ax:F3}{P3}%
$
    \InputIfFileExists{tapes/cipriano/axF3left.tikz}{}{\input{./tikz/tapes/cipriano/axF3left.tikz}}

 \raisebox{-0.6em}{\,\,$\stackrel{\eqref{ax:F3}}{=}$\,\,}
 
    \InputIfFileExists{tapes/cipriano/axF3right.tikz}{}{\input{./tikz/tapes/cipriano/axF3right.tikz}}
$
\\[1.2em]

\multicolumn{3}{c}{%
  \scalebox{0.85}{%
    \mylabelbis{ax:F4}{P4}%
    \mylabelbis{ax:F5}{P5}%
    $
    \InputIfFileExists{tapes/cipriano/axF4left.tikz}{}{\input{./tikz/tapes/cipriano/axF4left.tikz}}

      \,\stackrel{\eqref{ax:F4}}{=}\,
      
    \InputIfFileExists{tapes/cipriano/axF4right.tikz}{}{\input{./tikz/tapes/cipriano/axF4right.tikz}}

      \,\stackrel{\eqref{ax:F5}}{=}\,
      
    \InputIfFileExists{tapes/cipriano/axF5.tikz}{}{\input{./tikz/tapes/cipriano/axF5.tikz}}
$%
  }%
  \hspace{3em}%
  \mylabelbis{ax:F6}{P6}%
  $
    \InputIfFileExists{tapes/cipriano/axF6left.tikz}{}{\input{./tikz/tapes/cipriano/axF6left.tikz}}

    \,\stackrel{\eqref{ax:F6}}{=}\,
    
    \InputIfFileExists{tapes/cipriano/axF6right.tikz}{}{\input{./tikz/tapes/cipriano/axF6right.tikz}}
$%
  \hspace{3em}%
  \mylabelbis{ax:F7}{P7}%
  $
    \InputIfFileExists{tapes/cipriano/axF7left.tikz}{}{\input{./tikz/tapes/cipriano/axF7left.tikz}}

    \,\stackrel{\eqref{ax:F7}}{=}\,
    
    \InputIfFileExists{tapes/cipriano/axF7right.tikz}{}{\input{./tikz/tapes/cipriano/axF7right.tikz}}
$%
}

\end{tabular}
}
}
\caption{The monoidal theory of partial Boolean algebras.}
\label{tab:partialbooleanalgebra}
\end{figure}

The remainder of this section is devoted to proving the converse implication, namely completeness. 
Although this fact is not required for our proof, it is worth noting that the axiom \textsc{F8} in~\cite{piedeleu2025boolean}--corresponding to \eqref{ax:F8} in the following lemma--can be derived within $\ThPB$. 

\begin{lemma}\label{lemma:failureequalities}
  The following derived laws hold:
  \begin{center}
    \mylabelbis{ax:F9}{D6}%
    $
    \begin{tikzpicture}
	\begin{pgfonlayer}{nodelayer}
		\node [style=none] (176) at (-3.25, 0) {};
		\node [style=none] (178) at (-1.25, 0) {$\coflip{1}$};
		\node [style=none] (179) at (1.5, 0) {$\Flip{1}$};
	\end{pgfonlayer}
	\begin{pgfonlayer}{edgelayer}
		\draw (176.center) to (178.center);
	\end{pgfonlayer}
\end{tikzpicture}
}

      \stackrel{\eqref{ax:F9}}{\eqPB}
      
    \InputIfFileExists{tapes/cipriano/axF9right.tikz}{}{\input{./tikz/tapes/cipriano/axF9right.tikz}}
$
    \qquad
    \mylabelbis{ax:F8}{D7}%
    $
    \InputIfFileExists{tapes/cipriano/axF8left.tikz}{}{\input{./tikz/tapes/cipriano/axF8left.tikz}}

      \stackrel{\eqref{ax:F8}}{\eqPB}
      
    \InputIfFileExists{tapes/cipriano/axF8right.tikz}{}{\input{./tikz/tapes/cipriano/axF8right.tikz}}
$
    \qquad
    \mylabelbis{ax:F10l}{D8l}%
    \mylabelbis{ax:F10r}{D8r}%
    $
    \begin{tikzpicture}
	\begin{pgfonlayer}{nodelayer}
		\node [style=none] (0) at (-0.75, 0) {$\Flip{1}$};
		\node [style=none] (1) at (1, 0) {$\coflip{1}$};
	\end{pgfonlayer}
\end{tikzpicture}
}

      \stackrel{\eqref{ax:F10l}}{\eqPB}
      
    \InputIfFileExists{tapes/cipriano/axF10middle.tikz}{}{\input{./tikz/tapes/cipriano/axF10middle.tikz}}

      \stackrel{\eqref{ax:F10r}}{\eqPB}
      
    \begin{tikzpicture}
	\begin{pgfonlayer}{nodelayer}
		\node [style=none] (2) at (-0.75, 0) {$\Flip{0}$};
		\node [style=none] (3) at (1, 0) {$\coflip{0}$};
	\end{pgfonlayer}
\end{tikzpicture}
}
$
  \end{center}
\end{lemma}

The strategy for proving completeness is as follows. 
First, we show that any diagram $c$ in $\Diag{PB}$ can be decomposed into two Boolean circuits $D_c$ and $T_c$ in $\Diag{\SigB}$ (Proposition~\ref{prop:decompositionpartialcircuits}). 
We then appeal to the completeness result for Boolean circuits (Theorem~\ref{thm:completenessbooleancircuits}).
\begin{definition}
For all $c\in \Diag{PB}[A^n,A^m]$, the Boolean circuits $D_c\in \Diag{\SigB}[A^n,A]$ and $T_c\in\Diag{\SigB}[A^n,A^m]$ are inductively defined as
\begin{equation}\label{eq:defTD}  \begin{array}{rcl@{\hspace{1.5cm}}rcl}
    D_c &\defeq& \CBdischargern\, \Flip{1} &       T_c &\defeq& c \qquad \text{for all $c\in \SigB \cup \{\id{1},\id{A},\symmt{A}{A}\}$;}\\
    D_{\scalebox{0.6}{$\CBcocopier$}} &\defeq& 
    \InputIfFileExists{tapes/cipriano/xnorgatealone.tikz}{}{\input{./tikz/tapes/cipriano/xnorgatealone.tikz}}
 &     T_{\scalebox{0.6}{$\CBcocopier$}} &\defeq& \Andgate\\
    D_{c;d} &\defeq &\quad 
    \InputIfFileExists{tapes/cipriano/Dfg.tikz}{}{\input{./tikz/tapes/cipriano/Dfg.tikz}}
 &  T_{c;d} &\defeq &\quad 
    \InputIfFileExists{tapes/cipriano/Tfg.tikz}{}{\input{./tikz/tapes/cipriano/Tfg.tikz}}
\\
    D_{c\otimes d} &\defeq & \quad 
    \InputIfFileExists{tapes/cipriano/Dfperg.tikz}{}{\input{./tikz/tapes/cipriano/Dfperg.tikz}}
 & T_{c\otimes d} &\defeq& \quad 
    \InputIfFileExists{tapes/cipriano/Tfperg.tikz}{}{\input{./tikz/tapes/cipriano/Tfperg.tikz}}
\\
  \end{array}
  \end{equation}
\end{definition}
Intuitively, $D_c$ and $T_c$ represent, respectively, the domain and the total component of $c$. 
The diagram $D_c$ returns $1$ if the input of $c$ lies in the domain of definition of $c$, and $0$ otherwise. 
The diagram $T_c$ returns the output of $c$ when the input lies in its domain, and a vector of $0$ otherwise.
For instance, the domain of $\CBcocopier$ is the $\mathrm{xnor}$ gate: it returns $1$ if and only if the two inputs are equal. 
Its total part is given by the $\wedge$ gate: it returns $1$ precisely when both inputs are $1$.
Similarly, the domain of $c \otimes d$ is given by the conjunction of the domains of $c$ and $d$. 
Its total part returns the total parts of $c$ and $d$ when the domain evaluates to $1$, and $0$ otherwise.

\begin{proposition}\label{prop:decompositionpartialcircuits}
  For all $c\in \Diag{PB}$, it holds that {
    \InputIfFileExists{tapes/cipriano/decompositionpartial.tikz}{}{\input{./tikz/tapes/cipriano/decompositionpartial.tikz}}
}.
  \end{proposition}
\begin{proof}
  We proceed by induction on the structure of $c$. For the base case, if $c\in \SigB \cup \{\id{1},\id{A},\symmt{A}{A}\}$, then the statement follows by
  \begin{center}
    
    \InputIfFileExists{tapes/cipriano/propdecomposition1.tikz}{}{\input{./tikz/tapes/cipriano/propdecomposition1.tikz}}

  \end{center} 
  If $c$ is $\CBcocopier$, then the statement is exactly axiom \eqref{ax:F6} in Figure~\ref{tab:partialbooleanalgebra}. Now, in the case $c;d$, then assuming that the statement holds for $c$ and $d$, we have
  \begin{center}
    
    \InputIfFileExists{tapes/cipriano/propdecomposition2.tikz}{}{\input{./tikz/tapes/cipriano/propdecomposition2.tikz}}

    \end{center}
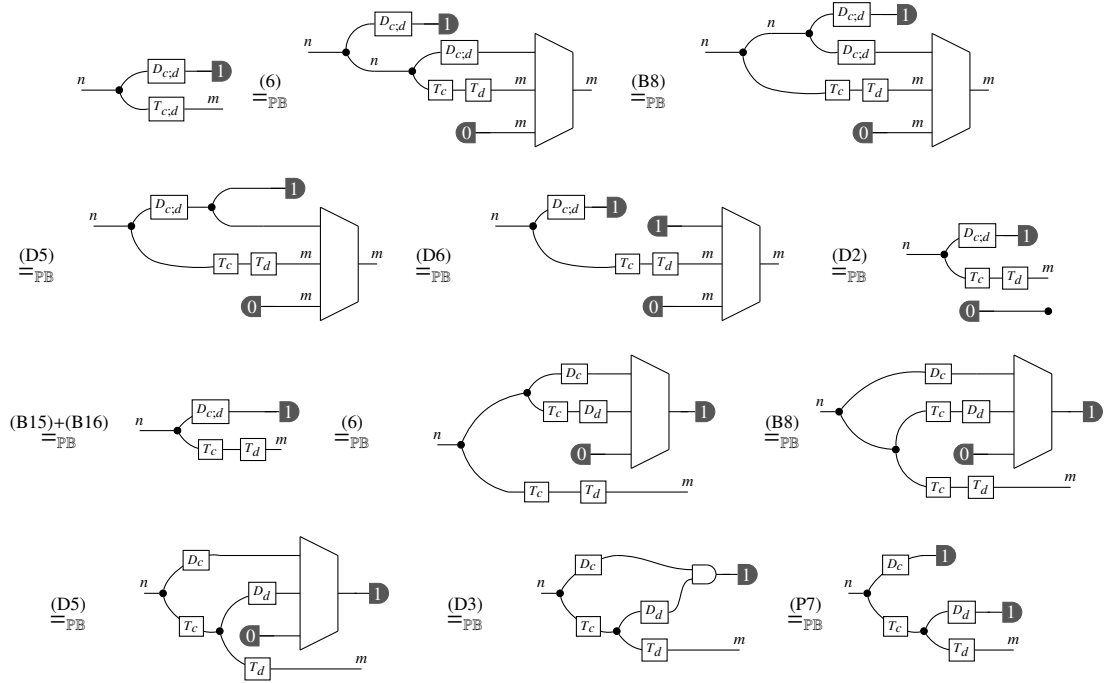
    \begin{center}
    
    \InputIfFileExists{tapes/cipriano/propdecomposition3.tikz}{}{\input{./tikz/tapes/cipriano/propdecomposition3.tikz}}

    \end{center}
    \begin{center}
    
    \InputIfFileExists{tapes/cipriano/propdecomposition4.tikz}{}{\input{./tikz/tapes/cipriano/propdecomposition4.tikz}}

    \end{center}
    \begin{center}
    
    \InputIfFileExists{tapes/cipriano/propdecomposition5.tikz}{}{\input{./tikz/tapes/cipriano/propdecomposition5.tikz}}

    \end{center}
    where now the statement follows by inductive hypothesis. Finally, the case $c\otimes d$ is left in Appendix~\ref{app:sec:probboolcircuits}.
\end{proof}


Now observe that, by Theorem~\ref{thm:completenessbooleancircuits}, every $b\in \Diag{\SigB}[1,A^n]$ is $\eqB$-equal (hence $\eqPB$-equal) to a Boolean vector, i.e., a circuit of the form $\bigotimes_{i=1}^n b_i$ for $b_i \in \{\Flip{0}, \Flip{1}\}$ if $n>0$, and $
    \InputIfFileExists{empty.tikz}{}{\input{./tikz/empty.tikz}}
$ if $n=0$. 

\begin{lemma}\label{lemma:complete1}
  Let $b\in \Diag{\SigB}[1,A^n]$ and $c\in \Diag{\SigPB}[A^n,A^m]$. If $b;D_c \eqPB \Flip{0}$, then $b;T_c\eqPB\Flipm{0}$.
\end{lemma}

\begin{lemma}\label{lemma:complete3}
  For all $b\in \Diag{\SigB}[1,A^n]$ and $c,d\in \Diag{PB}[A^n,A^m]$, 
  \begin{center} if $\osemPB{b;c}=\osemPB{b;d}$, then $\osemPB{b;D_c}=\osemPB{b;D_d}$ and $\osemPB{b;T_c}=\osemPB{b;T_d}$.\end{center}
\end{lemma}

\begin{theorem}[Completeness]\label{thm:completenesspartialcircuits}
  For all $c,d\in \Diag{PB}[A^n,A^m]$, if $\osemPB{c}=\osemPB{d}$ then $c=_{\mathbb{\SigPB }}d$.
\end{theorem}
\begin{proof}
  For all $b\in \Diag{\SigB}[1,A^n]$, functoriality of $\osemPB{-}$ and the hypothesis entail that
\[    \osemPB{b;c} =\osemPB{b};\osemPB{c} 
    =\osemPB{b};\osemPB{d} 
    =\osemPB{b;d}\text{.}\]
By  Lemma~\ref{lemma:complete3}, $\osemPB{b;D_c}=\osemPB{b;D_d}$ and $\osemPB{b;T_c}=\osemPB{b;T_d}$ for all $b\in \Diag{\SigB}[1,A^n]$. Hence $\osemPB{D_c} = \osemPB{D_d}$ and $\osemPB{T_c} = \osemPB{T_d}$. Now, since $D_c,D_d,T_c$ and $T_d$ are Boolean circuits in $\Diag{\SigB}$, then by Theorem~\ref{thm:completenessbooleancircuits} it holds that $D_c\eqB D_d$ and $T_c\eqB T_d$. Finally, by Proposition~\ref{prop:decompositionpartialcircuits}, since $c$ and $d$ decompose through the same domain and total part, we have that $c\eqPB d$.
\end{proof}
Moreover, one can characterise exactly the image of $\Diag{\SigPB}$ through $\osemPB{-}$. Let $\Cat{Par}_2$ be the full subcategory of $\Cat{Par}$ where objects are powers of the set $2$, i.e.,  $2^n$ for all $n \in \mathbb{N}$.

\begin{proposition}\label{prop:isopartialboolean}
The  functor  $\osemPB{-}\colon \Diag{{\SigPB }} \to \Cat{Par}$ factors as 
\[\xymatrix{ \Diag{\SigPB } \ar@{->>}[r]^{Q_{\mathbb{\SigPB }}}& \Diag{\mathbb{\SigPB }} \ar[r]^{\cong}& \Cat{Par}_2 \ar@{^{(}->}[r]& \Cat{Par} }\]
where the rightmost functors is the obvious injections and the central arrow is an isomorphism.
\end{proposition}

Theorem \ref{thm:completenesspartialcircuits} also provides a useful characterisation of diagrams in $\Diag{\SigPB}[1,A^n]$. We fix $\Flipn{\bot}\defeq \Flip{\bot};\ncopierbis\colon A \to A^n$ where $\ncopierbis\colon A \to A^n$ is  inductively defined as follows.
\begin{center}
  
    \InputIfFileExists{tapes/cipriano/ncopierbis.tikz}{}{\input{./tikz/tapes/cipriano/ncopierbis.tikz}}

\end{center}

\begin{lemma}\label{lemma:caratterizzazionecircuiti0n}
  For all $c\in \Diag{\SigPB}[1,A^n]$, either $c\eqPB \Flipn{\bot}$ or $c\eqPB b$ for some $b\in \Diag{\SigB}[1,A^n]$.
\end{lemma}

\section{Probabilistic Boolean circuits}\label{ssec:probabilisticboolean}
We now recall probabilistic Boolean circuits from \cite{piedeleu2025boolean}.
The  signature is $\SigPRB \defeq \SigPB \cup\{ \Flip{p} \mid p\in (0,1)\}$. The gate $\Flip{p}$ outputs a Boolean value, which is $1$ with probability $p$ and $0$ with probability $1-p$. 

To provide a formal semantics, it is convenient to recall some key notions. We write $\Dis(X)$ for the set of all (finitely supported probability) \emph{subdistributions} over a set  $X$, namely, functions $d\colon X\to [0,1]$  such that $\sum_{x\in X}d(x)\leq 1$, and $d(x)\neq 0$ for  finitely many $x\in X$. For all $x\in X$, we write $\delta_x \in \Dis(X)$ for the Dirac distribution at $x$; $\star\in \Dis(X)$ for the null subdistribution: $\star(x)\defeq 0$ for all $x\in X$; for all $p\in (0,1)$ and $d_1,d_2\in \Dis(X)$, we write $d_1+_pd_2$ for the subdistributions mapping each $x\in X$ into $p\cdot d_1(x) + (1-p)\cdot d_2(x)$.

The assignment $X\mapsto \Dis(X)$ extends to a monad $\subdistr\colon \Sets \to \Sets$ (see e.g.,~\cite{hasuo2007generic}). Its Kleisli category is denoted by $\KlD$: 
objects are sets; morphisms \(f \colon X \to Y\) are functions \(X \to \subdistr(Y)\). 
  We often write \(f(y \mid x)\) for \(f(x)(y)\), as this number represents the probability that \(f\) returns \(y\) given the input \(x\).
  Identities \(\id{X} \colon X \to \subdistr(X)\) map each \(x \in X\) to \(\delta_{x}\).
  For two functions \(f \colon X \to \subdistr(Y)\) and \(g \colon Y \to \subdistr(Z)\), their composition in $\KlD$ is defined as \(f ; g (z \mid x) \defeq \sum_{y \in Y} f(y \mid x) \cdot g(z \mid y)\).

This category carries a symmetric monoidal structure
$(\KlD, \otimes, \uno)$, where $\otimes$ is the cartesian product of sets and, for arrows  \(f \colon X_1 \to Y_1\) and \(g \colon X_2 \to Y_2\), $f\per g\colon X_1\times X_2 \to Y_1 \times Y_2$ is defined as 
\begin{equation}\label{ex:products}
\begin{aligned}
f \per g((y_1,y_2) \mid (x_1,x_2)) &\defeq f(y_1 \mid x_1) \cdot g(y_2 \mid x_2)\text{.}
\end{aligned}
\end{equation}
Finally, the symmetric monoidal functor $\ParJ\colon (\Cat{Par},\otimes , 1) \to (\KlD, \otimes ,1)$ is the identity on objects and associates to a partial function $f\colon X \to Y$ the function $\ParJ(f)\colon X \to \Dis(Y)$, mapping $x$ to $\delta_{f(x)}$, if $f(x)$ is defined and to the null distribution $\star$ otherwise.

\smallskip

The semantics is given by interpreting the generators in $\SigPRB$ as arrows in $\KlD$: $\Flip{p}$ as 
\begin{equation}\label{eq:semanticsFlip}
\begin{array}{rcl}
      \osem{\Flip{p}}\colon 1& \to & 2   \\
      \bullet&\mapsto& \delta_1 +_p \delta_0  \end{array}
    \end{equation}
while generators $s\in \SigPB$ as $\osem{s}\defeq \ParJ(\osemPB{s})$. Analogously to the previous sections, the inductive extension of the above interpretation yields a symmetric monoidal functor $\osem{-}\colon \Diag{\SigPRB} \to \KlD$, assigning to each diagram its semantics. For instance, $\osem{(\Flip{p} \otimes \Flip{p}) ; \Andgate} = (\osem{\Flip{p}} \otimes \osem{\Flip{p}}) \, ; \osem{\Andgate}$ which, by \eqref{eq:semanticsFlip} and \eqref{ex:products}, is the $\KlD$-arrows mapping $\bullet$ into 
$\delta_1 +_{p^2} \delta_0$, i.e., $\osem{\Flip{p^2}}$. Similarly, one can check that $\osem{\Flip{p} ; \CBcopier ; \Andgate}=\osem{\Flip{p}}$. This fact is also witnessed by \eqref{ax:B5} in \Cref{tab:booleanalgebra}.

A natural goal, mirroring the case of (partial) Boolean circuits, is to obtain a complete axiomatisation of semantic equality. Instead, we establish a complete axiomatisation for probabilistic Boolean tapes, a formalism into which circuits can be encoded.

\section{From circuits to tapes}\label{sec:encodingprobboolcircuits}

Example 30 in \cite{bonchi2025tapediagramsmonoidalmonads} encodes probabilistic Boolean circuits as \emph{probabilistic Boolean tapes}. 
This relies on $\KlD$ having finite coproducts, which induce a second monoidal structure $(\KlD, \piu, \zero)$, where $\piu$ is disjoint union of sets with units $\zero \defeq \{\}$;
for \(f \colon X_1 \to Y_1\) and \(g \colon X_2 \to Y_2\), the arrow \(f \piu g \colon X_1 \oplus X_2 \to Y_1 \oplus Y_2\) is defined, for all \(u \in X_1 \oplus X_2\), \(v \in Y_1 \oplus Y_2\), by
\begin{equation}\label{ex:coproducts}
\begin{aligned}
f\piu g(v \mid u) \defeq 
\begin{cases} 
f(y_1 \mid x_1) & \text{if } u=\iota_1(x_1) \text{ and } v=\iota_1(y_1),\\ 
g(y_2 \mid x_2) & \text{if } u=\iota_2(x_2)\text{ and } v=\iota_2(y_2), \\  
0 & \text{otherwise.}
\end{cases}
\end{aligned}
\end{equation}
Here $\iota_1$ and $\iota_2$ are the coproduct injections. In $\KlD$, $\otimes$ distributes over $\piu$, so $(\KlD, \piu, \per, \zero, \uno)$ forms a \emph{rig category}, aka bimonoidal category~\cite{laplaza_coherence_1972}. Moreover, for each object $X$, the arrows
\begin{equation}\label{ex:comonoids} 
\begin{array}{cccc} \arraycolsep=2pt \begin{array}{rcl} \diagp{X} \colon X & \to & X \oplus X \\ x & \mapsto & \delta_{\iota_1(x)}+_p\delta_{\iota_2(x)}
 \end{array} & \arraycolsep=2pt \begin{array}{rcl} \bang{X} \colon X & \to & \zero\\ x & \mapsto& \star \end{array} & \arraycolsep=2pt \begin{array}{rcl} \cobang{X}\colon \zero &\to& X\\ \text{ } \end{array} & \arraycolsep=2pt \begin{array}{rcl} \codiag{X} \colon X\oplus X & \to & X\\ \iota_i(x) & \mapsto& \delta_{x} \end{array} \end{array} \end{equation}
are central: $(\codiag{X}, \cobang{X})$ and $(\diagp{X}, \bang{X})$ form, respectively, a natural and coherent \emph{monoid} and \emph{co-pointed convex algebra}, namely they satisfy the laws in Tables~\ref{fig:freestrictfccat} and~\ref{fig:freecopcacat} in Appendix~\ref{app:sec:preliminaries}.

\smallskip

These structures are exploited in $\CatT{\Diag{\SigPB}}$, the category of probabilistic Boolean tapes. Consider the following context-free grammar, where $p \in (0,1)$, $U,V \in \sort^*$, and $c$ is an arrow in $\Diag{\SigPB}$:
\begin{equation}\label{tapesGrammar}
\setlength{\arraycolsep}{3pt}
\begin{array}{rcccccccccccccccccccc}
\t & ::= & \diagp{U} & \mid & \bangp{U} & \mid & \tapeFunct{c} & \mid & \cobang{U} & \mid & 
\codiag{U} & \mid & \id{U} & \mid & \id{\zero} & \mid & \sigma_{U,V}^{\piu} & \mid & \t ; \t & \mid & \t \piu \t
\end{array}
\end{equation}
We restrict to terms typable by the rules of Table~\ref{fig:freestricmmoncat-typing}. Types are arrows $P \to Q$, where $P,Q \in (\sort^*)^*$ are viewed as sums $\Piu[i=1][n] U_i$, with each $U_i = \Per[j=1][m_i] A$. 

The category $\CatT{\Diag{\SigPB}}$ has as set of objects $(\sort^*)^*$. Arrows in $\CatT{\Diag{\SigPB}}[P,Q]$ are terms of type $P\to Q$ modulo the axioms in Table~\ref{tabellone}, i.e., those of natural and coherent monoids, co-pcas and strict symmetric monoidal categories. 
By the latter axioms, $(\CatT{\Diag{\SigPB}}, \piu, \zero)$ is a symmetric monoidal category. 
Define $P \per Q \defeq \Piu[i]{\Piu[j]{U_i V_j}}$ for $P = \Piu[i]{U_i}$ and $Q = \Piu[j]{V_j}$; together with definition of $\t_1 \per \t_2$ in Table~\ref{tab:producttape} in Appendix~\ref{app:sec:preliminaries}, this yields another symmetric monoidal structure $(\CatT{\Diag{\SigPB}}, \otimes, \uno)$, which makes $(\CatT{\Diag{\SigPB}}, \piu, \per, \zero, \uno)$ a rig category. Crucially, there is a morphism of rig categories $\dsem{-} \colon \CatT{\Diag{\SigPB}} \to \KlD$ assigning to each tape its semantics in $\KlD$:
 \begin{equation}\label{eq:semtapes}
 \renewcommand{\arraystretch}{1.5}
\!\!\!\begin{tabular}{l@{\;\;\;\;}l@{\;\;\;\;}l@{\;\;\;\;}l@{\;\;\;\;}l}
$\CBdsem{\diagp{A^n}}  \defeq \; \diagp{2^n}$ & $\CBdsem{\bang{A^n}} \defeq \bang{2^n}$  & $\CBdsem{\, \tapeFunct{c} \,} \defeq \ParJ(\osemPB{c})  $ &  $\CBdsem{\cobang{A^n}}  \defeq \cobang{2^n}$    & $\CBdsem{\codiag{A^n}} \defeq \codiag{2^n} $ \\
$\CBdsem{\id{A^n}} \defeq \id{2^n} $&$ \CBdsem{\id{\zero}}  \defeq \id{\zero}  $&$ \CBdsem{\symmp{A^n}{A^m}} \defeq \symmp{2^n}{2^m} $&$ \CBdsem{\s;\t}  \defeq \CBdsem{\s}  ; \CBdsem{\t}  $&$ \CBdsem{\s \piu \t}  \defeq \CBdsem{\s}  \piu \CBdsem{\t}  $ 
\end{tabular}
\end{equation}

Arrows of $\CatT{\Diag{\SigPB}}$ have a diagrammatic representation. The grammar in \eqref{tapesGrammar} is depicted as
\begin{equation*}\label{tapesDiagGrammar} 
    \setlength{\tabcolsep}{2pt}
    \begin{tabular}{rc c@{$\,\mid\,$}c@{$\,\mid\,$}c@{$\,\mid\,$}c@{$\,\mid\,$}c@{$\,\mid\,$}c@{$\,\mid\,$}c@{$\,\mid\,$}c@{$\,\mid\,$}c@{$\,\mid\,$}c}
        $\t$ & $\Coloneqq$ & $
    \InputIfFileExists{/tapes/cipriano/pcomonoid.tikz}{}{\input{./tikz//tapes/cipriano/pcomonoid.tikz}}
$ & $\Tcounit{U}$  & $ \Tcirc{c}{U}{V}$ & $\Tunit{U}$  & $\Tmonoid{U}$ 
        & $\Twire{U}$ & $ 
    \InputIfFileExists{empty.tikz}{}{\input{./tikz/empty.tikz}}
 $    & $ \Tsymmp{U}{V} $ & $ 
    \InputIfFileExists{tapes/seq_comp.tikz}{}{\input{./tikz/tapes/seq_comp.tikz}}
  $ & $  
    \InputIfFileExists{tapes/par_comp.tikz}{}{\input{./tikz/tapes/par_comp.tikz}}
$  
    \end{tabular}
\end{equation*}  
Note that in $\Tcirc{c}{U}{V}$, $ \Cgen{c}{U}{V}$ is an arrow in $\Diag{\SigPB}$, i.e., a string diagrams. Hence, string diagrams occur within tapes. Moreover, vertical composition of string diagrams corresponds to $\per$, while vertical composition of tapes to $\piu$. For instance, $\tapeFunct{\Flip{1} \per \Flip{0}}\colon 1 \to A\per A$ and $\tapeFunct{\Flip{1}} \piu \tapeFunct{\Flip{0}}\colon \uno \piu \uno \to A \piu A$ are drawn respectively as 
    \InputIfFileExists{tapes/cipriano/tapeunoperzero.tikz}{}{\input{./tikz/tapes/cipriano/tapeunoperzero.tikz}}
 and 
    \InputIfFileExists{tapes/cipriano/tapeunopiuzero.tikz}{}{\input{./tikz/tapes/cipriano/tapeunopiuzero.tikz}}
.
The graphical representation embodies several axioms of Table \ref{tabellone}. 
Those  axioms that are not implicit in the graphical representation are displayed in Figure~\ref{fig:tapesax}.

\begin{figure}[t]
\centering
\setlength{\tabcolsep}{2pt}
\renewcommand{\arraystretch}{0.9}
\resizebox{0.92\linewidth}{!}{%
\begin{tabular}{r@{\;}c@{\;}l @{\qquad} r@{\;}c@{\;}l @{\qquad} r@{\;}c@{\;}l}
    \mylabelbis{ax:tapes:symminv}{\ensuremath{\symmp\text{-inv}}}%
    
    \InputIfFileExists{tapes/ax/symminv_left.tikz}{}{\input{./tikz/tapes/ax/symminv_left.tikz}}
 &$\stackrel{\eqref{ax:tapes:symminv}}{=}$& 
    \InputIfFileExists{tapes/ax/symminv_right.tikz}{}{\input{./tikz/tapes/ax/symminv_right.tikz}}

    &
    \mylabelbis{ax:tapes:symmnat}{\ensuremath{\symmp\text{-nat}}}%
    
    \InputIfFileExists{tapes/ax/symmnat_left.tikz}{}{\input{./tikz/tapes/ax/symmnat_left.tikz}}
 &$\stackrel{\eqref{ax:tapes:symmnat}}{=}$& 
    \InputIfFileExists{tapes/ax/symmnat_right.tikz}{}{\input{./tikz/tapes/ax/symmnat_right.tikz}}

    &
    \mylabelbis{ax:tapes:sigmainv}{\ensuremath{\sigma\text{-inv}}}%
    
    \InputIfFileExists{cb/symm_inv_left.tikz}{}{\input{./tikz/cb/symm_inv_left.tikz}}
 &$\stackrel{\eqref{ax:tapes:sigmainv}}{=}$& 
    \InputIfFileExists{cb/symm_inv_right.tikz}{}{\input{./tikz/cb/symm_inv_right.tikz}}

    \\
    \mylabelbis{ax:tapes:codiagas}{\ensuremath{\codiag{}\text{-as}}}%
    
    \InputIfFileExists{tapes/whiskered_ax/monoid_assoc_left.tikz}{}{\input{./tikz/tapes/whiskered_ax/monoid_assoc_left.tikz}}
 &$\stackrel{\eqref{ax:tapes:codiagas}}{=}$& 
    \InputIfFileExists{tapes/whiskered_ax/monoid_assoc_right.tikz}{}{\input{./tikz/tapes/whiskered_ax/monoid_assoc_right.tikz}}

    &
    \mylabelbis{ax:tapes:codiagun}{\ensuremath{\codiag{}\text{-un}}}%
    
    \InputIfFileExists{tapes/whiskered_ax/monoid_unit_left.tikz}{}{\input{./tikz/tapes/whiskered_ax/monoid_unit_left.tikz}}
 &$\stackrel{\eqref{ax:tapes:codiagun}}{=}$& \Twire{U}
    &
    \mylabelbis{ax:tapes:codiagsym}{\ensuremath{\codiag{}\text{-sym}}}%
    
    \InputIfFileExists{tapes/whiskered_ax/monoid_comm_left.tikz}{}{\input{./tikz/tapes/whiskered_ax/monoid_comm_left.tikz}}
 &$\stackrel{\eqref{ax:tapes:codiagsym}}{=}$& \Tmonoid{U}
    \\
    \mylabelbis{ax:tapes:diagpas}{\ensuremath{\diagp{}{}\text{-as}}}%
    
    \InputIfFileExists{tapes/cipriano/monoid_assoc_left.tikz}{}{\input{./tikz/tapes/cipriano/monoid_assoc_left.tikz}}
 &$\stackrel{\eqref{ax:tapes:diagpas}}{=}$& 
    \InputIfFileExists{tapes/cipriano/monoid_assoc_right.tikz}{}{\input{./tikz/tapes/cipriano/monoid_assoc_right.tikz}}

    &
    \mylabelbis{ax:tapes:diagpidem}{\ensuremath{\diagp{}\text{-idem}}}%
    \scalebox{0.8}{
    \InputIfFileExists{tapes/cipriano/idempotency.tikz}{}{\input{./tikz/tapes/cipriano/idempotency.tikz}}
} &$\stackrel{\eqref{ax:tapes:diagpidem}}{=}$& \Twire{U}
    &
    \mylabelbis{ax:tapes:diagpsym}{\ensuremath{\diagp{}{}\text{-sym}}}%
    
    \InputIfFileExists{tapes/cipriano/monoid_comm_left.tikz}{}{\input{./tikz/tapes/cipriano/monoid_comm_left.tikz}}
 &$\stackrel{\eqref{ax:tapes:diagpsym}}{=}$& 
    \InputIfFileExists{tapes/cipriano/pcomonoidbar.tikz}{}{\input{./tikz/tapes/cipriano/pcomonoidbar.tikz}}

    \\
    \multicolumn{9}{c}{%
        \scalebox{0.85}{%
        \begin{tabular}{r@{\;}c@{\;}l @{\quad} r@{\;}c@{\;}l @{\quad} r@{\;}c@{\;}l @{\quad} r@{\;}c@{\;}l}
            \mylabelbis{ax:tapes:diagpnat}{\ensuremath{\diagp{}\text{-nat}}}%
            
    \InputIfFileExists{tapes/cipriano/pcanatleft.tikz}{}{\input{./tikz/tapes/cipriano/pcanatleft.tikz}}
 &$\stackrel{\eqref{ax:tapes:diagpnat}}{=}$& 
    \InputIfFileExists{tapes/cipriano/pcanatright.tikz}{}{\input{./tikz/tapes/cipriano/pcanatright.tikz}}

            &
            \mylabelbis{ax:tapes:codiagnat}{\ensuremath{\codiag{}\text{-nat}}}%
            
    \InputIfFileExists{tapes/cipriano/monoidnatnewright.tikz}{}{\input{./tikz/tapes/cipriano/monoidnatnewright.tikz}}
 &$\stackrel{\eqref{ax:tapes:codiagnat}}{=}$& 
    \InputIfFileExists{tapes/cipriano/monoidnatnewleft.tikz}{}{\input{./tikz/tapes/cipriano/monoidnatnewleft.tikz}}

            &
            \mylabelbis{ax:tapes:bangnat}{\ensuremath{\,\bangp{}\text{-nat}}}%
            
    \InputIfFileExists{tapes/cipriano/bangnatleft.tikz}{}{\input{./tikz/tapes/cipriano/bangnatleft.tikz}}
 &$\stackrel{\eqref{ax:tapes:bangnat}}{=}$& 
    \InputIfFileExists{tapes/cipriano/bangnatright.tikz}{}{\input{./tikz/tapes/cipriano/bangnatright.tikz}}

            &
            \mylabelbis{ax:tapes:cobangnat}{\ensuremath{\,\cobang{}\text{-nat}}}%
            
    \InputIfFileExists{tapes/cipriano/cobangnatleft.tikz}{}{\input{./tikz/tapes/cipriano/cobangnatleft.tikz}}
 &$\stackrel{\eqref{ax:tapes:cobangnat}}{=}$& 
    \InputIfFileExists{tapes/cipriano/cobangnatright.tikz}{}{\input{./tikz/tapes/cipriano/cobangnatright.tikz}}

        \end{tabular}
        }%
    }
\end{tabular}
}%
\caption{Axioms for probabilistic tape diagrams, where $\tilde{p}\defeq pq$ and $\tilde{q}\defeq \frac{p(1-q)}{1-pq}$ }
\label{fig:tapesax}
\end{figure}

Consider the tape $\diagp{1}; (\tapeFunct{\Flip{1}} \piu \tapeFunct{\Flip{0}}) ; \codiag{A} \colon \uno \to A$ which is drawn as the diagram below.
\begin{equation}\label{eq:exprobabilistictapes}
\Flip{p}[t] \defeq 
    \InputIfFileExists{tapes/examples/1p0.tikz}{}{\input{./tikz/tapes/examples/1p0.tikz}}

\end{equation}
By \eqref{eq:semtapes},  its semantics $\CBdsem{\Flip{p}[t]}$ is  $\diagp{1}; (\, \ParJ(\osemPB{\Flip{1}}) \piu \ParJ(\osemPB{\Flip{0}}) \,) ; \codiag{2}$ which --by \eqref{eq:sembool},  \eqref{ex:coproducts} and \eqref{ex:comonoids}-- is the function $\bullet \mapsto \delta_1 +_p \delta_0$ . Note that this is exactly $\osem{\Flip{p}}$ as defined in \eqref{eq:semanticsFlip}. This simple observation is at the core of the encoding $\encoding{-} \colon (\Diag{\SigPRB},\otimes, \uno) \to (\CatT{\Diag{\SigPB}},\otimes, \uno)$ from \cite{bonchi2025tapediagramsmonoidalmonads} reported in Table \ref{tab:encoding}.

\begin{table}
\[\begin{array}{rcll|rcl}
\encoding{c}&\defeq & \tapeFunct{c} & \text{ for }c\in \SigPB\cup\{\id{A},\id{1},\symmt{A}{A}\} & \encoding{c;d}&\defeq & \encoding{c};\encoding{d}\\
\encoding{\Flip{p}} & \defeq & \Flip{p}[t] & &\encoding{c \per d}&\defeq & \encoding{c} \per \encoding{d}\\
\end{array}\]
\caption{The encoding $\encoding{-} \colon (\Diag{\SigPRB},\otimes, \uno) \to (\CatT{\Diag{\SigPB}},\otimes, \uno)$}\label{tab:encoding}
\end{table}

\begin{proposition}\label{prop:encoding} For all $c\in \Diag{\SigPRB}$, $\osem{c}=\CBdsem{\encoding{c}}$.
\end{proposition}

The above result asserts that the encoding $\encoding{-}$ preserves the semantics. Thus, instead of axiomatising the equivalence induced by the semantics on probabilistic Boolean circuits, we rather axiomatise the one on tapes (Corollary \ref{cor:completenessPBPtapes}) and then check equivalence of circuits by encoding them into tapes (Corollary \ref{cor:finale}). Before illustrating the axiomatision for tapes and prove its completeness in Section~\ref{sec:completenessprobbooltapes}, we  recall some key results from \cite{probbooltapesfossacs} in the next section. We concluse this section by recalling from \cite{bonchi2025tapediagramsmonoidalmonads} an example comparing the expressivity of circuits and tapes.

\begin{example}\label{ex:esempioSTRONG}
In~\cite{piedeleu2025boolean}, probabilistic control is realised via a multiplexer, represented by the tape $\scalebox{0.7}{\Ifgate[t]}$. Intuitively, when $\Flip{p}[t]$, $\Tcirc{c}{\!\!\!\!}{\!\!\!\!}$ and $\Tcirc{d}{\!\!\!\!}{\!\!\!\!}$ are connected, respectively, to the first, second, and third inputs of the multiplexer, the resulting output coincides with that of $c$ with probability $p$ and with that of $d$ with probability $1-p$.
Formally, this behaviour is captured by the composite $ (\Flip{p}[t] \;\per\; \Tcirc{c}{\!\!\!\!}{\!\!\!\!} \;\per\; \Tcirc{d}{\!\!\!\!}{\!\!\!\!}) \; ; \; \scalebox{0.7}{\Ifgate[t]}$
which, by the definition of $\Flip{p}[t]$ in~\eqref{eq:exprobabilistictapes} and of $\per$ in Table~\ref{tab:producttape}, corresponds to the tape shown on the left below.
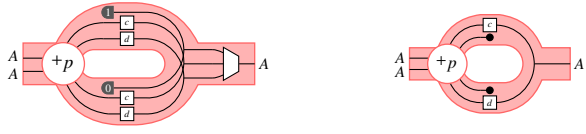
\begin{equation*}
    
    \InputIfFileExists{tapes/examples/multiplexT.tikz}{}{\input{./tikz/tapes/examples/multiplexT.tikz}}

    \qquad \qquad
    
    \InputIfFileExists{tapes/examples/pchoice.tikz}{}{\input{./tikz/tapes/examples/pchoice.tikz}}

\end{equation*}
Although the two diagrams above exhibit similar behaviour, a crucial difference emerges when $d$ (or symmetrically $c$) is instantiated as $\Flip{\bot}[t]$. In this case, the circuit on the left always fails: see \eqref{ax:F8}. By contrast, the circuit on the right still produces the output of $c$ (respectively $d$) with probability $p$ (respectively $1-p$). Since the composition of any behaviour with a null behaviour via $\otimes$ necessarily results in a null behaviour, we believe that issues analogous to the one above are inherent in approaches relying exclusively on the monoidal product $\otimes$. The introduction of $\oplus$, on the other hand, offers a natural and expressive means of modelling probabilistic control.
\end{example}

    \section{On Freely Generated Convex Biproduct Categories}\label{sec:preliminaries}

The category $\CatT{\Diag{\SigPB}}$ of probabilistic Boolean tapes is an instance of a general construction, introduced in~\cite{arxivprobbooltapes}, which associates to any category $\Cat{C}$ a category $\CatTapeC$. 
The latter is defined analogously to $\CatT{\Diag{\SigPB}}$, but with $\sort^*$ replaced by $\ob{\Cat{C}}$ and arrows of $\Diag{\SigPB}$ replaced by arrows of $\Cat{C}$. Consequently, the objects of $\CatTapeC$ are elements of $\ob{\Cat{C}}^*$, and its arrows are equivalence classes of well-typed terms generated by the grammar in~\eqref{tapesGrammar}, where now $U,V \in \ob{\Cat{C}}$ and $c$ ranges over arrows of $\Cat{C}$.

In this section, we briefly recall from~\cite{arxivprobbooltapes} the properties of $\CatTapeC$ that are needed for the completeness proof. 
Throughout, all categories are tacitly assumed to be locally small, and $\Cat{Cat}$ denotes the category of locally small categories.

\subsection{PCA-enriched categories}\label{sub:sec:preliminaries:pca}
Recall  from~\cite{stone1949postulates,bonchi2017power} that a  \emph{pointed convex algebra} (pca) consists of a set $X$, a designated element $\star \in X$ and, for all $p$ in the open real interval $(0,1)$, a function $+_p \colon X\times X \to X$ such that, by fixing $\tilde{p}\defeq pq$ and $\tilde{q}\defeq \frac{p(1-q)}{1-pq}$, the following laws hold for all $x_1,x_2,x_3\in X$. 
\begin{equation}\label{eq:pca}(x_1+_q x_2)+_p x_3 =  x_1 +_{\tilde{p}} (x_2+_{\tilde{q}} x_3) \qquad x_1+_px_2=x_2+_{1-p}x_1 \qquad  x_1+_p x_1 = x_1\end{equation}
We denote by $\Cat{PCA}$ the category of pcas and their morphisms, i.e., functions preserving $\star$ and $+_p$.

In any pca, $+_p$ is easily extended to any $p\in[0,1]$ by fixing $x+_1 y \defeq x$ and $x+_0 y \defeq y$.
 This  can be further extended for any $n\in \mathbb{N}$ and $p_1,\dots,p_n\in [0,1]$ inductively as 
    $\sum_{i=1}^0 p_i\cdot (-)_i \defeq \star$ and   $\sum_{i=1}^{n+1} p_i\cdot (-)_i \defeq (-)_{1} +_{p_{1}} \sum_{j=1}^n q_j \cdot (-)_j$ 
where $(-)_j=(-)_{i+1}$ and $q_j$ is $\frac{p_{i+1}}{1-p_1}$ if $p_1 \neq 1$ and $0$ otherwise. For $n=1$, this is the \emph{multiplication by a scalar} $p\in [0,1]$, shortly defined as $p\cdot x\defeq x+_p \star$.

The archetypical example of a pca is  $\Dis(X)$ with $+_p$ and $\star$  defined as in Section \ref{ssec:probabilisticboolean}. Such pca enjoys an additional property called cancellativity. A pca $(X,+_p,\star)$ is \emph{cancellative} (or, in the terminology of \cite{sokolova2018termination}, cancellative at $\star$) if  for all $x,y\in X$ and $p\in(0,1)$:
$p\cdot x = p\cdot y \Rightarrow x=y$.

    A category $\Cat{C}$ is \emph{$\Cat{PCA}$-enriched} if every homset carries a pca structure and composition of arrows is a pcas morphism, namely that the following equalities hold for all properly typed arrows $e,f,g,h$.
\begin{equation}\label{eq:enr}e; (f+_p g) = (e;f)+_p (e;g) \qquad (f+_pg) ;h= (f;h +_p g;h) \qquad f;\star =\star= \star;f\end{equation}
A functor $F$ is $\Cat{PCA}$-enriched if it preserves the pca structure of each homset. $\Cat{PCA}$-enriched categories and functors form a category, denoted by $\Cat{PCACat}$.

Moreover, a monoidal category $(\Cat{C},\otimes, 1)$ is \emph{monoidally enriched}  over $\Cat{PCA}$ if the following holds.
\begin{equation}\label{eq:monoidalenrichment} 
e \otimes (f+_p g)= (e \otimes f)+_p(e \otimes g) \qquad  (f+_p g)  \otimes h= ( f \otimes h )+_p(g \otimes  h) \qquad \star \otimes\, f= \star = f\otimes \star 
\end{equation}
 The category $\KlD$ is monoidally enriched over $\Cat{PCA}$: for all $f,g\colon X \to Y$, $x\in X$, $y\in Y$ and $p\in (0,1)$,
$f +_p g(y|x)\defeq p\cdot f(y|x) + (1-p)\cdot g(y|x)$ and $\star_{X,Y}(y|x) \defeq 0$.  
Also $\CatTapeC$ is $\Cat{PCA}$-enriched: $f +_p g \defeq  \,\diagp{X}; (f\oplus g); \codiag{Y}$ and $\star_{X,Y} \defeq \bang{X};\cobang{Y}$.
Such enrichment is monoidal whenever $(\Cat{C}, \otimes, 1)$ is a symmetric monoidal category: see \cite[Thm.\ 27]{bonchi2025tapediagramsmonoidalmonads}.


\subsection{Convex Biproduct Categories}\label{sub:sec:preliminaries:cbc}
We now recall the definition of convex biproduct categories from \cite{arxivprobbooltapes}, which are suitable $\Cat{PCA}$-enriched categories where coproducts enjoy an additional universal property, similar to the one of products.

\begin{definition}\label{def:convprod}
Let $X_1$, $X_2$  be two objects of a $\Cat{PCA}$-enriched category $\Cat{C}$. The \emph{convex product} of $X_1$ and $X_2$ is an object $Z$ with two arrows $\pi_1\colon Z\to X_1$ and $\pi_2\colon Z\to X_2$ satisfying the following property: for all $p_1,p_2 \in [0,1]$ such that $p_1+p_2\le 1$ and all arrows $f_1\colon A\to X_1$, $f_2\colon A \to X_2$, there exists a unique arrow $h \colon A \to Z$ such that $h;\pi_1 = p_1\cdot f$ and $h;\pi_2 = p_2 \cdot g$. 
\end{definition}

Similarly, the convex product of $n$ objects $X_1,\ldots,X_n$ is an object $Z$ with arrows $\pi_i\colon Z\to X_i$ for $i=1,\ldots,n$ satisfying the following property: for all $p_1, \dots, p_n \in [0,1]$ where $\sum_{i=1}^n p_i \le 1$ and arrows $f_i\colon A\to X_i$, there exists a unique arrow $h\colon A \to Z$ such that $h;\pi_i = p_i\cdot f_i$ for all $i=1,\ldots,n$. Observe that, by definition, the 0-ary convex product is a final object. Hereafter, we will denote the unique arrow $h$ by $\langle f_1, \dots, f_n \rangle_{\vec{p}}$ where $\vec{p}$ is a compact notation for $p_1, \dots , p_n$.

\begin{definition}\label{def:convbicat}
A \emph{convex biproduct category} is a $\Cat{PCA}$-enriched category $\Cat{C}$ with an object $\zero$ which is both initial and final and, for every pair of objects $X_1,X_2$, an object $X_1\oplus X_2$ and morphisms $\pi_i \colon X_1\oplus X_2 \to X_i$ and $\iota_i \colon X_i \to X_1 \oplus X_2$ such that $(X_1\oplus X_2, \iota_1, \iota_2)$ is a coproduct, $(X_1\oplus X_2, \pi_1, \pi_2)$ is a convex product and $\iota_i; \pi_j = \id{X_i}$ if $i=j$ and $\iota_i; \pi_j = \star_{X_i,X_j}$ if $i \neq j$.

A morphism of convex biproduct categories is a $\Cat{PCA}$-enriched functor $F\colon \Cat{C} \to \Cat{D}$ preserving finite coproducts. We write $\Cat{CBCat}$ for the category of convex biproduct categories and their morphisms.
\end{definition}
Note that the above definition is obtained from that of \emph{category with finite biproducts} \cite{mac_lane_categories_1978,coecke2017two}  by just replacing products by convex products. Examples of convex biproduct categories include  $\KlD$, its continuous analogue (see \cite{arxivprobbooltapes}) and, most importantly, $\CatTapeC$.
\begin{theorem}\label{thm:TCconvexbiproductcategory}
$\CatTapeC$ is a convex biproduct category. In particular, for every arrow $\t\colon U \to \bigoplus_{i=1}^n U_i$ in $\CatTapeC$, there exist $n$ arrows $\t_i\colon U \to U_i$ such that $\t = \langle \t_1, \dots \t_n\rangle_{\vec{p}}$ for some $\vec{p}=p_1,\dots ,p_n$.
\end{theorem}

Crucially, the assignment $\Cat{C} \mapsto \CatTapeC$ extends to a functor $\CatTapeFUN\colon \Cat{Cat} \to \Cat{CBCat}$ which is left adjoint to the forgetful functor  $U\colon \Cat{CBCat} \to \Cat{Cat} $. 
\begin{theorem}\label{thm:syntacticadjunction}
$\CatTapeFUN\colon \Cat{Cat} \to \Cat{CBCat}$ is left adjoint to $U\colon \Cat{CBCat} \to \Cat{Cat}$.
\end{theorem}
In other words, $\CatTapeC$ is the convex biproduct category freely generated by $\Cat{C}$. Theorems 1 and 2 in \cite{probbooltapesfossacs} illustrate that such adjunction is decomposed in the following two adjunctions, where $\Cat{PCACat}$ denotes the category of $\Cat{PCA}$-enriched categories and $\Cat{PCA}$-enriched functors: 
\begin{equation}\label{eq:adjunctions}
\xymatrix{
\Cat{Cat} \ar@/^/[rr]^{(-)^+}& \;\;{\tiny{\bot}} & \Cat{PCACat}  \ar@/^/[ll]^{U} \ar@/^/[rr]^{\stmat{-}} &\;{\tiny{\bot}}  & \Cat{CBCat} \ar@/^/[ll]^{U}
}
\end{equation}
\begin{corollary}\label{cor:isotapematrices}
For all categories $\Cat{C}$, $\CatTapeC$ is isomorphic to $\stmat{\Cat{C}^+}$ in $\Cat{CBCat}$.
\end{corollary}
While it is not strictly necessary for our completeness proof, it is useful to provide some intuition about the two constructions in \eqref{eq:adjunctions}.

The leftmost adjunction is an instance of a general result about enriching over arbitrary algebraic theories: see \cite[Prop. 6.4.7]{borceux2} or \cite[Cor. 1]{villoria2024enriching}. For every category $\Cat{C}$ one obtains $\Cat{PCA}$-enriched category $\Cat{C}^+$: objects of  $\Cat{C}^+$ are those of $\Cat{C}$; for all objects $X,Y$, the homset is defined as $\Cat{C}^+[X,Y] \defeq \Dis(\Cat{C}[X,Y])$. For $d_1\colon X \to Y$ and $d_2\colon Y \to Z$, their composition $d_1;d_2 \colon X\to Z$ is defined for all $h\in \Cat{C}[X,Z]$ as $d_1;d_2(h) \defeq \sum_{\{(f,g)\mid f;g=h\}}d_1(f) \cdot d_2(g)$; The identity $\id{X}\colon X \to X$ is given by $\delta_{id_X}$.

The second adjunction is an adaptation of \cite[Exercises VIII.2.5-6]{mac_lane_categories_1978} which constructs a category of matrices from a category enriched over commutative monoids. 
Given a $\Cat{PCA}$-enriched category $\Cat{C}$, $\stmat{\Cat{C}}$ is the category of \emph{stochastic matrices} over $\Cat{C}$.
In a nutshell, objects of $\stmat{\Cat{C}}$ are words in $Ob(\Cat{C})^*$, while arrows $M\colon\bigoplus_{k=1}^n U_k\to \bigoplus_{k=1}^m V_k$ are matrices with $(j,i)$-entries given by pairs $ (p_{ji}, f_{ji})$ where $f_{ji}\in\Cat{C}[U_i,V_j]$ and $p_{ji}\in[0,1]$ satisfy $\sum_{j=1}^{m}p_{ji}\le 1$. The composition of two morphisms is given by matrix multiplication where addition is given by the pca enrichment of $\Cat{C}$ and multiplication by  composition of arrows in $\Cat{C}$.

        \section{A Complete Axiomatisation for Probabilistic Boolean Tapes}\label{sec:completenessprobbooltapes}
Corollary~\ref{cor:isotapematrices} informs us that $\CatT{\Diag{\SigPB}}$ is isomorphic to $\stmat{\Diag{\SigPB}^+}$. That is, probabilistic Boolean tapes are in bijective correspondence with stochastic matrices whose entries are subdistributions over diagrams in \(\Diag{\SigPB}\). This correspondence--summarised in Figure~\ref{fig:dictionary}--can be illustrated by the example below, where a tape (on the left) is paired with its associated matrix (on the right):
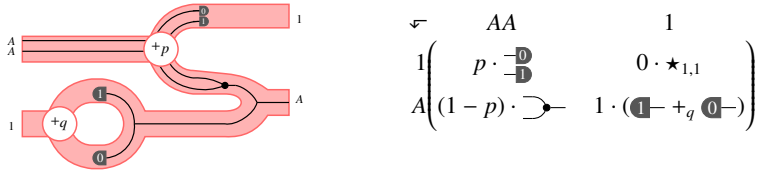
\begin{equation*}\label{eq:tapeintro1}
\resizebox{!}{3.2em}{$
    \InputIfFileExists{tapes/cipriano/composizione2booleana.tikz}{}{\input{./tikz/tapes/cipriano/composizione2booleana.tikz}}
$}
\qquad
\qquad
\raisebox{-0.4ex}{$
\begin{pNiceMatrix}[first-col,first-row]
	\rotatebox{90}{$\Lsh$} & AA & 1 \\
	1 & p\cdot 
    \begin{tikzpicture}
	\begin{pgfonlayer}{nodelayer}
		\node [scale=0.8, style=none] (0) at (0, -0.5) {$\coflip{1}$};
		\node [scale=0.8, style=none] (1) at (0, 0.5) {$\coflip{0}$};
	\end{pgfonlayer}
\end{tikzpicture}
}
 & 0\cdot \star_{1,1} \\
	A & (1-p)\cdot \CBcocopier & 1\cdot (\Flip{1}+_q \Flip{0})
\end{pNiceMatrix}
$}
\end{equation*}
We now consider the axioms in \(\ThPB\) and the category \(\Diag{\ThPB}\), obtained as the quotient of \(\Diag{\SigPB}\) by \(\ThPB\). By applying Corollary~\ref{cor:isotapematrices} with \(\Cat{C} = \Diag{\ThPB}\), we obtain that \(\CatT{\Diag{\ThPB}}\) is isomorphic to \(\stmat{\Diag{\ThPB}}\). By Proposition~\ref{prop:isopartialboolean}, the latter is in turn isomorphic to \(\stmat{\Cat{Par}_2}\). Altogether:
\begin{equation}\label{eq:i2so}
\CatT{\Diag{\ThPB}} \cong \stmat{\Diag{\ThPB}^+} \cong \stmat{\Cat{Par}_2^+}.
\end{equation}

However, \(\stmat{\Cat{Par}_2^+}\) does not embed faithfully into \(\KlD\): morphisms in \(\stmat{\Cat{Par}_2^+}\) carry strictly more information than those in \(\KlD\). 
To see this, consider the tapes \(\scalebox{0.6}{
    \InputIfFileExists{tapes/cipriano/AXPBP2right.tikz}{}{\input{./tikz/tapes/cipriano/AXPBP2right.tikz}}
}\) and \(\Flip{\bot}[t]\), together with their corresponding matrices in \(\stmat{\Cat{Par}_2^+}\):
\[
\raisebox{-0.4ex}{$
\begin{pNiceMatrix}[first-col,first-row]
	\rotatebox{90}{$\Lsh$} & 1  \\
	2 & 0\cdot \star_{1,2} 
\end{pNiceMatrix}
$}
\qquad
\raisebox{-0.4ex}{$
\begin{pNiceMatrix}[first-col,first-row]
	\rotatebox{90}{$\Lsh$} & 1  \\
	2 & 1\cdot \osemPB{\Flip{\bot}}
\end{pNiceMatrix}
$}
\]
The unique entry in the left matrix is the null subdistribution \(\star \in \Dis(\Cat{Par}[1,2])\), whereas the entry in the right matrix is the Dirac distribution concentrated on the partial function \(\osemPB{\Flip{\bot}} \colon 1 \to 2\). These are two distinct morphisms in \(\Cat{Par}_2^+[1,2]\). 
Nevertheless, one readily verifies that \(\dsem{\scalebox{0.6}{
    \InputIfFileExists{tapes/cipriano/AXPBP2right.tikz}{}{\input{./tikz/tapes/cipriano/AXPBP2right.tikz}}
}}\) and \(\dsem{\Flip{\bot}[t]}\) coincide as morphisms in \(\KlD\): both correspond to the arrow \(\star_{1,2} \colon 1 \to 2\), i.e., the function \(1 \to \Dis(2)\) sending the unique element \(\bullet \in 1\) to the null subdistribution \(\star \in \Dis(2)\).

\begin{figure}[t]
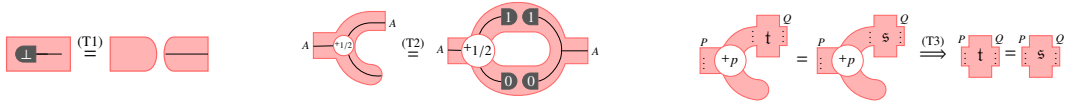

  \centering

  \mylabelbis{ax:AXPBP1}{T2}
  \mylabelbis{ax:AXPBP2}{T1}
  \mylabelbis{ax:canc}{T3}

  \[
  \begin{tabular*}{\textwidth}{@{\extracolsep{\fill}}ccc@{}}
    \scalebox{0.90}{$
      
    \InputIfFileExists{tapes/cipriano/AXPBP2left.tikz}{}{\input{./tikz/tapes/cipriano/AXPBP2left.tikz}}

      \stackrel{\eqref{ax:AXPBP2}}{=}
      
    \InputIfFileExists{tapes/cipriano/AXPBP2right.tikz}{}{\input{./tikz/tapes/cipriano/AXPBP2right.tikz}}

    $}
    &
    \scalebox{0.80}{$
      
    \InputIfFileExists{tapes/cipriano/AXPBP1left.tikz}{}{\input{./tikz/tapes/cipriano/AXPBP1left.tikz}}

      \stackrel{\eqref{ax:AXPBP1}}{=}
      
    \InputIfFileExists{tapes/cipriano/AXPBP1right.tikz}{}{\input{./tikz/tapes/cipriano/AXPBP1right.tikz}}

    $}
    &
    \scalebox{0.80}{$
      
    \InputIfFileExists{tapes/cipriano/Axcancleft.tikz}{}{\input{./tikz/tapes/cipriano/Axcancleft.tikz}}

      \stackrel{\eqref{ax:canc}}{\implies}
      
    \InputIfFileExists{tapes/cipriano/Axcancright.tikz}{}{\input{./tikz/tapes/cipriano/Axcancright.tikz}}

    $}
  \end{tabular*}
  \]

  \caption{Axioms for tapes of partial Boolean circuits.}
  \label{ax:BooleanTAPES}
\end{figure}

\subsection{Axiomatisation}\label{ssec:probbooltapes:axiomatisation}
To obtain a complete axiomatisation, we need to add to $\ThPB$ the axioms in Figure~\ref{ax:BooleanTAPES}. 

First, the axiom \eqref{ax:AXPBP2} forces the equality of $\Flip{\bot}[t]$ and \(\scalebox{0.6}{
    \InputIfFileExists{tapes/cipriano/AXPBP2right.tikz}{}{\input{./tikz/tapes/cipriano/AXPBP2right.tikz}}
}\). Then, consider the tape diagram on the right-hand side of \eqref{ax:AXPBP1}: with probability $\frac{1}{2}$ it tests if the input is $1$ and returns $1$ and, with probability $\frac{1}{2}$, tests if the input is $0$ and returns $0$. The equality \eqref{ax:AXPBP1} forces such tape to be equal to $\frac{1}{2} \cdot \id{A}$.  Finally, axiom \eqref{ax:canc} is an implication which asserts cancellativity: see Section~\ref{sub:sec:preliminaries:pca}. 

From the theory $\ThPB$ together with the axioms in Figure~\ref{ax:BooleanTAPES}, we generate a congruence $\dot{\sim}$ on $\CatT{\Diag{\SigPB}}$ via the inference rules below.
\begin{equation}\label{eq:congr3}
        \scalebox{0.85}{$
        \centering
        \begin{array}{c}
        \begin{array}{c@{\qquad\qquad}c@{\qquad\qquad}c@{\qquad\qquad}c@{\qquad\qquad}c}
            \inferrule*[right=\eqref{ax:AXPBP2}]{\t_1 \stackrel{\eqref{ax:AXPBP2}}{=} \t_2}{\t_1 \dot{\sim} \t_2}
            &
            \inferrule*[right=\eqref{ax:AXPBP1}]{\t_1 \stackrel{\eqref{ax:AXPBP1}}{=} \t_2}{\t_1 \dot{\sim} \t_2}
            &
            \inferrule*[right=($\textsc{R}$)]{-}{\t \dot{\sim} \t}
            &
            \inferrule*[right=($\textsc{S}$)]{\t_1 \dot{\sim} \t_2}{\t_2 \dot{\sim} \t_1}
            &    
            \inferrule*[right=($\textsc{T}$)]{\t_1 \dot{\sim} \t_2 \quad \t_2 \dot{\sim} \t_3}{\t_1 \dot{\sim} \t_3}
        \end{array}
        \\[10pt]
        \begin{array}{c@{\qquad}c@{\qquad}c@{\qquad}c@{\qquad}c}
        \inferrule*[right=($\ThPB$)]{c \eqPB d}{\tape{c} \dot{\sim} \tape{d}}
            &
            \inferrule*[right=\eqref{ax:canc}]{p\cdot \t \dot{\sim} p\cdot \s }{\t \dot{\sim} \s}
            &
            \inferrule*[right=($;$)]{\t_1 \dot{\sim} \t_2 \quad \s_1 \dot{\sim} \s_2}{\t_1;\s_1 \dot{\sim} \t_2;\s_2}
            &
            \inferrule*[right=($\piu$)]{\t_1 \dot{\sim} \t_2 \quad \s_1 \dot{\sim} \s_2}{\t_1\piu\s_1 \dot{\sim} \t_2 \piu \s_2}
            &
            \inferrule*[right=($\per $)]{\t_1 \dot{\sim} \t_2 \quad \s_1 \dot{\sim} \s_2}{\t_1\per \s_1 \dot{\sim} \t_2 \per \s_2}       
        \end{array}
        \end{array}
        $}
\end{equation}

A straightforward verification shows that the axioms are sound.

\begin{proposition}[Soundness]\label{prop:soundnesstapes}
For all $\s,\t \in \CatT{\Diag{\SigPB}}$, if $\s \eqsynbis \t$ then $\dsem{\s}=\dsem{\t}$.
\end{proposition}

For the completeness proof, it is convenient to work with $\CatT{\Diag{\ThPB}}$ rather than $\CatT{\Diag{\SigPB}}$. 
To this end, consider the functor
$
\CatT{Q_{\ThPB}} \colon \CatT{\Diag{\SigPB}} \to \CatT{\Diag{\ThPB}},
$
obtained by applying $\CatT{-}$ to the quotient functor $Q_{\ThPB}\colon \Diag{\SigPB} \to \Diag{\ThPB}$, i.e., the canonical functor identifying diagrams modulo $=_{\ThPB}$.
Moreover, let $\eqsyn$ denote the congruence on $\CatT{\Diag{\ThPB}}$ defined as $\dot{\sim}$ but without the rule $(\ThPB)$. The following result shows that reasoning with $\eqsyn$ suffices.

\begin{proposition}\label{prop:twoequivalence}
For all $\s , \t \in \CatT{\Diag{\SigPB}}$, $\s \dot{\sim} \t$ iff $\CatT{Q_{\ThPB}}(\s) \eqsyn \CatT{Q_{\ThPB}}(\t)$.
\end{proposition}


Henceforth, we write $\Cat{B}[1,A^n]$ for the set of $n$-ary Boolean vectors, i.e., the set of tapes in $\CatT{\Diag{\ThPB}}[1,A^n]$ of the form
$
\overrightarrow{b} \defeq \bigotimes_{i=1}^n b_i,
$
where each $b_i \in \{\Flip{0}[t], \Flip{1}[t]\}$ if $n \neq 0$, and $\overrightarrow{b} \defeq \id{\uno}$ when $n = 0$.
Given $\overrightarrow{b} \in \Cat{B}[1,A^n]$, we denote by $\overleftarrow{b}$ the tape
$
\bigotimes_{i=1}^n b'_i \in \CatT{\Diag{\ThPB}}[A^n,1],
$
where $b'_i$ is defined as follows: $b'_i = \coflip{1}[t]$ if $b_i = \Flip{1}[t]$, and $b'_i = \coflip{0}[t]$ if $b_i = \Flip{0}[t]$.

We can now illustrate three key properties of $\eqsyn$. First, axiom \eqref{ax:AXPBP2} easily entails the following.

\begin{lemma}\label{lemma:bottomstar2}
 $\star_{A^0,A^n} \sim \Flipn{\bot}[t]$.  
\end{lemma}

The second property crucially relies on axiom \eqref{ax:AXPBP1} and on the monoidal enrichment in \eqref{eq:monoidalenrichment}.

\begin{lemma}\label{lemma:axiomsninput}
 For all $n\in \mathbb{N}$, $\frac{1}{2^n}\cdot \id{A^n} \sim \sum_{\overrightarrow{b}\in \Cat{B}[1,A^n]}^{} \frac{1}{2^n}\cdot\overleftarrow{b};\overrightarrow{b}$;
\end{lemma}

Note that for $n=1$, the lemma above reduces to axiom \eqref{ax:AXPBP1}. Combining Lemma~\ref{lemma:axiomsninput} and axiom \eqref{ax:canc} we obtain the third key result.

\begin{lemma}\label{lemma:PBP2} 
Let $\s,\t\in \CatT{\Diag{\ThPB}}[A^n,A^m]$. If, for all $\overrightarrow{b}\in \Cat{B}[1,A^n]$, $\overrightarrow{b};\s \eqsyn\overrightarrow{b}; \t$, then $\s \eqsyn \t$.
\end{lemma}

\subsection{Completeness}\label{ssec:probbooltapes:completeness}
%

We denote by $\eqsynq{\CatT{\Diag{\ThPB}}}$ the quotient of $\CatT{\Diag{\ThPB}}$ by $\eqsyn$, and by
$
Q_{\eqsyn} \colon \CatT{\Diag{\ThPB}} \to \eqsynq{\CatT{\Diag{\ThPB}}}
$
the corresponding quotient functor. By Propositions~\ref{prop:soundnesstapes} and \ref{prop:twoequivalence}, $\dsem{-}\colon \CatT{\Diag{\SigPB}} \to \KlD$ factors as
\[\xymatrix{\CatT{\Diag{\SigPB}} \ar[r]^{\CatT{Q_\ThPB}} & \CatT{\Diag{\ThPB}} \ar[r]^{Q_\eqsyn} & \eqsynq{\CatT{\Diag{\ThPB}}} \ar@{.>}[r]^I& \KlD}\]
for some functor $I$. Proving completeness amounts to showing that $I$ is faithful.


\medskip
 Now, 
observe that $\ParJ\colon \,\Cat{Par} \to \KlD$ from Section~\ref{ssec:partialboolean} restricts to $\ParJ_2\colon \Cat{Par}_2 \to \KlD$. Since $\KlD$ is a convex biproduct category, by the two adjunctions in \eqref{eq:adjunctions}, there exists a functor \[\ParJ_2'\colon \stmat{\Cat{Par}_2^+} \to \KlD\text{.} \]
\vspace{-0.4cm}
\begin{minipage}[c]{0.55\textwidth}
Consider the diagram on the right, where the vertical isomorphisms are those in \eqref{eq:i2so}. 
Such diagram commutes, as stated by the following result where $\comp$ is the composition of the left vertical isomorphisms with $\ParJ_2'$.
\begin{lemma}\label{lemma:diagramma-commutativo}
$I\circ \eqsynq{Q}=\comp$.
\end{lemma}
\end{minipage}
\begin{minipage}[t]{0.15\textwidth}
\raisebox{0.7cm}{$
\qquad \xymatrix@R=2ex@C=3ex{
\CatT{\Diag{\ThPB}} \ar[r]^{\eqsynq{Q}} \ar[d]_{\cong} 
  & \eqsynq{\CatT{\Diag{\ThPB}}} \ar[dd]^{I} \\
{\stmat{\Diag{\ThPB}^+}} \ar[d]_{\cong} 
  & {} \\
\stmat{\Cat{Par}_2^+} \ar[r]_{\ParJ_2'} 
  & \KlD
}
$}
\end{minipage}
\vspace{0.4cm}


We first observe that $\comp$ is faithful for subdistributions of Boolean vectors.
\begin{lemma}\label{lemma:Ffaith}
For all $\t,\t'\in \CatT{\Diag{\ThB}}[A^0,A^m]$, if $\comp(\t)=\comp(\t')$, then $\t=\t'$.
\end{lemma}

Moreover, by Lemma~\ref{lemma:caratterizzazionecircuiti0n}, we have the following normal forms for tapes $\t \colon A^0 \to A^m$.

\begin{lemma}\label{lemma:PBP0}
For all  $\t\in \CatT{\Diag{\ThPB}}[A^0,A^m]$, there exist, $p,q\in [0,1]$ and for $i=1,\dots,n$,  $p_i\in (0,1)$ such that $\sum_{i=1}^{n}p_i= 1$, and $\overrightarrow{b}_i\in \Cat{B}[1,A^m]$ such that $$\t=(\sum_{i=1}^{n}p_i\cdot \overrightarrow{b}_i)+_p(\star_{A^0,A^m} +_q \Flipm{\bot}[t])\text{.}$$ 
  \end{lemma}

The previous two results, combined with Lemma~\ref{lemma:bottomstar2} allow us to prove the following key property.

\begin{lemma}\label{lemma:PBP1}
  For all $\t,\t'\in \CatT{\Diag{\ThPB}}[A^0,A^m]$, if $\comp(\t)=\comp(\t')$, then $\t\sim \t'$.
\end{lemma}

\begin{theorem}\label{thm:completenessPBPtapes}
  $I\colon \eqsynq{\CatT{\Diag{\ThPB}}} \to \stmat{\KlD_2}$ is faithful.
\end{theorem}
\begin{proof}
  Let $[\t]_\sim, [\s]_\sim \in \eqsynq{\CatT{\Diag{\ThPB}}}[A^n,A^m]$ be such that $I([\t]_\sim)=I([\s]_\sim)$. Then: 
  \begin{align}
    I([\t]_\sim)=I([\s]_\sim) &\implies \comp(\t)=\comp(\s) \tag{Lemma \ref{lemma:diagramma-commutativo}}\\
    &\implies \text{for all } \overrightarrow{b}\in \Cat{B}[1,A^n], \comp(\overrightarrow{b});\comp(\t)=\comp(\overrightarrow{b});\comp(\s) \notag \\
    &\implies \text{for all } \overrightarrow{b}\in \Cat{B}[1,A^n], \comp(\overrightarrow{b};\t)=\comp(\overrightarrow{b};\s) \tag{Functoriality}\\
    &\implies \text{for all } \overrightarrow{b}\in \Cat{B}[1,A^n], \overrightarrow{b};\t \eqsyn \overrightarrow{b};\s \tag{Lemma \ref{lemma:PBP1}}\\
    &\implies \t \eqsyn \s \tag{Lemma \ref{lemma:PBP2}}
  \end{align}
  Hence, $[\t]_\sim=[\s]_\sim$. The above derivation proves that $I$ is faithful for arrows of type $A^n \to A^m$. For arrows of type $ A^k \to \bigoplus_{i=1}^nA^{m_i}$ one can rely on convex products (Definition \ref{def:convprod}) and reduce to the previous case. While for arrows of arbitrary type $\bigoplus_{j=1}^o A^{k_j} \to \bigoplus_{i=1}^nA^{m_i}$, one can easily rely on the universal property of coproducts. See Appendix~\ref{app:{sec:completenessprobbooltapes-NEW2}} for all details.
\end{proof}

\begin{corollary}[Completeness]\label{cor:completenessPBPtapes}
  For all $\s,\t \in \CatT{\Diag{\SigPB}}$, if $\dsem{\s}=\dsem{\t}$ then $\s \eqsynbis \t$.
\end{corollary}

\begin{corollary}\label{cor:finale}
  For all $c,d \in \Diag{\SigPRB}$, $\osem{c}=\osem{d}$ iff $\encoding{c} \eqsynbis \encoding{d}$.
\end{corollary}

    \section{Conclusions and Future Work}

We introduced an axiomatisation of \emph{partial} Boolean circuits (\Cref{tab:booleanalgebra,tab:partialbooleanalgebra}) and proved it complete (\Cref{thm:completenesspartialcircuits}). Extending this, we gave laws of \emph{probabilistic Boolean tapes} (\Cref{fig:tapesax,ax:BooleanTAPES}) and, using a general result from~\cite{probbooltapesfossacs}, established their completeness (\Cref{cor:completenessPBPtapes}). Combined with the encoding in~\cite{bonchi2025tapediagramsmonoidalmonads}, this allows checking semantic equality of \emph{probabilistic Boolean circuits}~\cite{piedeleu2025boolean} (\Cref{cor:finale}).


\begin{wrapfigure}{r}{0.3\textwidth}
\vspace{-1em}
\centering

    \InputIfFileExists{tapes/cipriano/tapesproposalphase2/dim0.tikz}{}{\input{./tikz/tapes/cipriano/tapesproposalphase2/dim0.tikz}}

\vspace{-1em}
\end{wrapfigure}
In Example \ref{ex:esempioSTRONG}, it is shown that the combined use of $\oplus$ and $\otimes$ allows probabilistic tapes to realise probabilistic control effectively. Extending tapes with traces for $\oplus$--following the approach of~\cite{bonchi2024diagrammatic}, which provides a unifying account of several program logics~\cite{hoare1969axiomatic,o2019incorrectness,DBLP:journals/corr/abs-2310-18156,DBLP:conf/vmcai/CousotCFL13}--would enable the representation of iteration and, potentially, provide algebraic foundations for probabilistic program logics~\cite{kaminski2019advanced,olmedo2016reasoning}. 
Indeed, such extended tapes can encode probabilistic regular expressions from~\cite{DBLP:conf/lics/RozowskiS24}. For instance, the expression $i;c^{[p]};t$ is depicted as the tape above.
%
%
%

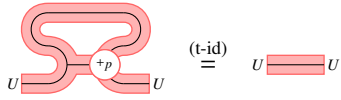
\begin{wrapfigure}{l}{0.34\textwidth}
\centering
$
    \InputIfFileExists{tapes/cipriano/tapesproposalphase2/replhs.tikz}{}{\input{./tikz/tapes/cipriano/tapesproposalphase2/replhs.tikz}}

\; \stackrel{\text{(t-id)}}{=} \;

    \InputIfFileExists{tapes/cipriano/tapesproposalphase2/reprhs.tikz}{}{\input{./tikz/tapes/cipriano/tapesproposalphase2/reprhs.tikz}}
$
\end{wrapfigure}
\noindent This extension is technically feasible since, as shown in~\cite{jacobs2010coalgebraic}, $\KlD$ is a traced monoidal category satisfying the \emph{uniformity principle}~\cite{hasegawa2003uniformity} as well as the axiom on the right. 
Combining this axiom with the laws introduced in this paper may lead to a proof principle analogous to \emph{martingales}.

  \bibliography{references}

\appendix
\section{Additional figures and tables}\label{app:sec:preliminaries}

\begin{table}[H]
\centering

\begin{subtable}{\textwidth}
\centering
\resizebox{\textwidth}{!}{%
$
\begin{array}{c}
\toprule
\begin{array}{c}
    \begin{array}{ccccc}
        {\diagp{U}\colon U \to U\oplus U} &
        {\bang{U} \colon U \to \zero} &
        {\sigma_{U, V}^{\piu} \colon U \piu V \to V \piu U} &
        {\codiag{U}\colon U \piu U \to U} &
        {\cobang{U} \colon \zero \to U}   
    \end{array}    
\\[0.3em]
    \begin{array}{ccccc}
        {id_\zero \colon \zero \to \zero} &
        {id_U \colon U \to U} &
        \inferrule{c \colon U \to V}{\tapeFunct{c}\colon U \to V} &
        \inferrule{\t \colon P \to Q \and \s \colon Q \to R}{\t ; \s \colon P \to R} &
        \inferrule{\t \colon P_1 \to Q_1 \and \s \colon P_2 \to Q_2}{\t \piu \s \colon P_1 \piu P_2 \to Q_1 \piu Q_2}       
    \end{array}
\end{array}
\\
\bottomrule
\end{array}
$
}
\caption{Typing rules.}
\label{fig:freestricmmoncat-typing}
\end{subtable}
\vspace{1em}

\begin{subtable}{\textwidth}
\centering
\resizebox{\textwidth}{!}{%
$
\begin{array}{c}
\toprule
\begin{array}{cc}
\begin{array}{ccc}
 \diagp{\zero}\defeq \id{\zero} &\quad
 \bang{\zero}\defeq \id{\zero} &\quad
 \bang{U \oplus P}\defeq \bang{U}\oplus \bang{P}
\end{array}
&
\begin{array}{ccc}
 \codiag{\zero}\defeq \id{\zero} &\quad
 \cobang{\zero}\defeq \id{\zero} &\quad
 \cobang{U \oplus P}\defeq \cobang{U}\oplus \cobang{P}
\end{array}
\end{array}
\\[1em]
\begin{array}{cc}
\diagp{U \oplus P}\defeq (\diagp{U} \piu\, \diagp{P});(\id{U} \piu \symm{U}{P} \piu \id{P})
&
\codiag{U\oplus P}\defeq (\id{U} \piu \symm{P}{U} \piu \id{P});(\codiag{U}\piu \codiag{P})
\end{array}
\\
\bottomrule
\end{array}
$
}
\caption{Inductive definition of monoid and co-pca structures.}
\label{eq:indmoncopca}
\end{subtable}

\vspace{1em}

\begin{subtable}{\textwidth}
\centering
\resizebox{\textwidth}{!}{%
$
\begin{array}{c}
\toprule
\begin{array}{cc}
\begin{array}{c}
(f;g);h=f;(g;h) \qquad id_P;f=f=f;id_Q\\
(f_1\piu f_2) ; (g_1 \piu g_2) = (f_1;g_1) \piu (f_2;g_2)
\end{array} 
&
\begin{array}{c}
id_{\zero}\piu f = f = f \piu id_{\zero} \qquad (f \piu g)\, \piu h = f \piu \,(g \piu h) \\
\sigma_{P, Q}; \sigma_{Q, P}= id_{P \piu Q} \qquad (\gen \piu id_R) ; \sigma_{Q, R} = \sigma_{P,R} ; (id_R \piu \gen)
\end{array}
\\[0.5em]
\multicolumn{2}{c}{
\tapeFunct{\id{P}}=\id{P} 
\qquad
\tapeFunct{c;d} = \tapeFunct{c}; \tapeFunct{d} \qquad (\text{Tape})
}
\end{array}
\\
\bottomrule
\end{array}
$
}
\caption{Axioms for strict symmetric monoidal categories and Tape axioms.}
\label{fig:freestricmmoncat-axioms} 
\end{subtable}

\vspace{1em}
\begin{subtable}{\textwidth}
\centering
\resizebox{\textwidth}{!}{%
\begin{tabular}{cc cc}
    \toprule
    $(\id{ P}\piu \codiag{ P}) ; \codiag{ P} = (\codiag{ P}\piu \id{ P}) ; \codiag{ P}$ & (\newtag{$\codiag{}$-as}{eq:codiag assoc}) &  $(\cobang{ P}\piu \id{ P}) ; \codiag{ P}  = \id{ P}  $ & (\newtag{$\codiag{}$-un}{eq:codiag unital}) \\[0.3em]
    $\cobang{\zero} = \id{\zero} \qquad\codiag{\zero} = \id{\zero}$ & (\newtag{$\cobang{\zero},\codiag{\zero}$-coh}{eq:codiag zero coherence}) & $\sigma_{ P, P};\codiag{ P}=\codiag{ P}$ & (\newtag{$\codiag{}$-sym}{eq: codiag symmetry}) \\[0.3em]
    $\cobang{P \piu Q} = \cobang{P} \piu \cobang{Q}$ & (\newtag{\,$\cobang{}$-coh}{eq:cobang coherence}) & $\codiag{P \piu Q} = (\id{P} \piu \sigma_{Q,P} \piu \id{Q}) ; (\codiag{P}\piu \codiag{Q}) $ & (\newtag{$\codiag{}$-coh}{eq:codiag coherence}) \\[0.3em]
    $\cobang{P};f =\cobang{Q}$ & (\newtag{\,$\cobang{}$-nat}{eq:cobang nat}) & $\codiag{P};f =(f\piu f); \codiag{Q}$ & (\newtag{$\codiag{}$-nat}{eq:codiag nat}) \\
    \bottomrule
\end{tabular}}
\caption{Axioms for natural and coherent monoids in a strict monoidal category.}
\label{fig:freestrictfccat}
\end{subtable}

\vspace{1em}

\begin{subtable}{\textwidth}
\centering
\resizebox{\textwidth}{!}{%
\begin{tabular}{cc cc}
    \toprule
    $\diagp{P};(\diagq{P}\piu \id{P})=\, \diagptilde{P};(\id{P}\piu \,\diagqtilde{})$ & (\newtag{$\diagp{}$-as}{eq:diagp assoc}) &  $\tilde{p}\defeq pq\qquad \tilde{q}\defeq \frac{p(1-q)}{1-pq}$ &  \\[0.3em]
    $\diagp{P};\codiag{P}  = \id{ P}$ & (\newtag{$\diagp{}$-idem}{eq:diagp idempotency}) & $\diagp{P}\sigma_{ P, P}=\,\diagpbar{P}$ & (\newtag{$\diagp{}$-sym}{eq:diagp symmetry}) \\[0.3em]
    $\diagp{\zero} = \id{\zero}$ & (\newtag{$\diagp{\zero}$-coh}{eq:diagp zero coherence}) & $\bangp{\zero} = \id{\zero}$ & (\newtag{\,$\bangp{\zero}$-coh}{eq:bangp0 coherence}) \\[0.3em]
    $\bangp{P \piu Q} = \bangp{P} \piu\, \bangp{Q}$ & (\newtag{\,$\bangp{}$-coh}{eq:bangp coherence}) & $\diagp{P \piu Q} = (\diagp{P}\piu\, \diagp{Q});(\id{P} \piu \sigma_{P,Q} \piu \id{Q}) $ & (\newtag{$\diagp{}$-coh}{eq:diagp coherence}) \\[0.3em]
    $f; \bangp{Q}=\bangp{P}$ & (\newtag{\,$\bangp{}$-nat}{eq:bangp nat}) & $f; \diagp{Q}=\,\diagp{P}; (f \piu f)$ & (\newtag{$\diagp{}$-nat}{eq:diagp nat}) \\
    \bottomrule
\end{tabular}}
\caption{Axioms for natural and coherent co-pcas in a strict monoidal category.}
\label{fig:freecopcacat}
\end{subtable}

\caption{Type system, syntactic sugar and axioms for $\CatT{\Cat{C}}$.}
\label{tabellone}
\end{table}

\begin{table}[H]
    \begin{center}
    {
        \hfill {\tiny
  \[\begin{array}{c}
  \toprule
        \def\arraystretch{1.2}
        \begin{array}{cc}
            \begin{array}{@{}l}
                \dl{P}{Q}{R} \colon P \per (Q\piu R)  \to (P \per Q) \piu (P\per R) \vphantom{\symmt{P}{Q}} \\
                \midrule
                \dl{\zero}{Q}{R} \defeq \id{\zero} \vphantom{\symmt{P}{\zero} \defeq \id{\zero}} \qquad
                \dl{U \piu P'}{Q}{R} \defeq (\id{U\per (Q \piu R)} \piu \dl{P'}{Q}{R});(\id{U\per Q} \piu \symmp{U\per R}{P'\per Q} \piu \id{P'\per R}) \vphantom{\symmt{P}{V \piu Q'} \defeq \dl{P}{V}{Q'} ; (\Piu[i]{\tapesymm{U_i}{V}} \piu \symmt{P}{Q'})} \\
            \end{array}
            &
            \begin{array}{@{}l}
                \symmt{P}{Q} \colon P\per Q \to Q \per P, \text{ with } P = \Piu[i]{U_i} \\
                \midrule
                \symmt{P}{\zero} \defeq \id{\zero} \qquad
                \symmt{P}{V \piu Q'} \defeq \dl{P}{V}{Q'} ; (\Piu[i]{\tapesymm{U_i}{V}} \piu \symmt{P}{Q'})
                \phantom{\quad}
            \end{array}
        \end{array}\\
        \midrule
        \begin{array}{rclrcl|rclrcl}
            \midrule
            \LW U {\id\zero} &\defeq& \id\zero& \LW U {\t_1 \piu \t_2} &\defeq& \LW U {\t_1} \piu \LW U {\t_2} &  \RW U {\id\zero} &\defeq& \id\zero & \RW U {\t_1 \piu \t_2} &\defeq& \RW U {\t_1} \piu \RW U {\t_2} \\
            \LW U {\tape{c}} &\defeq& \tape{\id U \per c} &    \LW U {\t_1 ; \t_2} &\defeq& \LW U {\t_1} ; \LW U {\t_2}&  \RW U {\tape{c}} &\defeq& \tape{c \per \id U} & \RW U {\t_1 ; \t_2} &\defeq& \RW U {\t_1} ; \RW U {\t_2} \\
            \LW U {\symmp{V}{W}} &\defeq& \symmp{UV}{UW}       &       && & \RW U {\symmp{V}{W}} &\defeq& \symmp{VU}{WU}&  &&  \\
                        \LW U {\diagp V} &\defeq& \diagp{UV} &  \LW U {\bang V} &\defeq& \bang{UV}& \RW U {\diagp V} &\defeq& \diagp{VU} & \RW U {\bang V} &\defeq& \bang{VU} \\
            \LW U {\codiag V} &\defeq& \codiag{UV} &  \LW U {\cobang V} &\defeq& \cobang{UV}& \RW U {\codiag V} &\defeq& \codiag{VU} & \RW U {\cobang V} &\defeq& \cobang{VU} \\
            \hline \hline
            \LW{\zero}{\t} &\defeq& \id{\zero} &  \LW{W\piu S'}{\t} &\defeq& \LW{W}{\t} \piu \LW{S'}{\t} & \RW{\zero}{\t} &\defeq& \id{\zero} &
            \RW{W \piu S'}{\t} &\defeq& \dl{P}{W}{S'} ; (\RW{W}{\t} \piu \RW{S'}{\t}) ; \Idl{Q}{W}{S'} \\
            \hline \hline
            \multicolumn{12}{c}{
                \t_1 \per \t_2 \defeq \LW{P}{\t_2} ; \RW{S}{\t_1}   \quad \text{ ( for }\t_1 \colon P \to Q, \t_2 \colon R \to S   \text{ )}
            }
            \\
            \bottomrule
        \end{array}
    \end{array}        
      \]
      }
      \hfill
      \caption{Inductive definitions of  left distributor $\dl{P}{Q}{R}$ and $\otimes$-symmetry $\symmt{P}{Q}$ (top); monomial and polynomial whiskerings (center); definition of $\per$ (bottom).}\label{tab:producttape}
       }
    \end{center}
\end{table}

\begin{figure}
\centering
\begingroup
\setlength{\tabcolsep}{4pt}
\renewcommand{\arraystretch}{1.05}
\small
\begin{tabular}{c c c}
\toprule
Symbol & Diagram & Matrix\\
\midrule

$\diagp{U}$
&

    \InputIfFileExists{/tapes/cipriano/pcomonoid.tikz}{}{\input{./tikz//tapes/cipriano/pcomonoid.tikz}}

&
\raisebox{-0.4ex}{$
\begin{pNiceMatrix}[first-col,first-row]
\rotatebox{90}{$\Lsh$} & U \\
U & p\cdot \wire{U} \\
U & (1-p)\cdot \wire{U}
\end{pNiceMatrix}
$}
\\[0.2em]

$\bangp{U}$
&
\Tcounit{U}
&
\raisebox{-0.4ex}{$
\begin{pNiceMatrix}[first-col,first-row]
\rotatebox{90}{$\Lsh$} \\
U
\end{pNiceMatrix}
$}
\\[0.2em]

$\tapeFunct{c}$
&
\Tcirc{c}{U}{V}
&
\raisebox{-0.4ex}{$
\begin{pNiceMatrix}[first-col,first-row]
\rotatebox{90}{$\Lsh$} & U \\
V & 1\cdot \Cgen{c}{U}{V}
\end{pNiceMatrix}
$}
\\[0.2em]

$\cobang{U}$
&
\Tunit{U}
&
\raisebox{-0.4ex}{$
\begin{pNiceMatrix}[first-col,first-row]
\rotatebox{90}{$\Lsh$} & U
\end{pNiceMatrix}
$}
\\[0.2em]

$\codiag{U}$
&
\Tmonoid{U}
&
\raisebox{-0.4ex}{$
\begin{pNiceMatrix}[first-col,first-row]
\rotatebox{90}{$\Lsh$} & U & U \\
U & 1\cdot \wire{U} & 1\cdot \wire{U}
\end{pNiceMatrix}
$}
\\[0.2em]

$\id{U}$
&
\Twire{U}
&
\raisebox{-0.4ex}{$
\begin{pNiceMatrix}[first-col,first-row]
\rotatebox{90}{$\Lsh$} & U \\
U & 1\cdot \wire{U}
\end{pNiceMatrix}
$}
\\[0.2em]

$\id{\zero}$
&

    \InputIfFileExists{empty.tikz}{}{\input{./tikz/empty.tikz}}

&
\raisebox{-0.4ex}{$
\begin{pNiceMatrix}[first-col,first-row]
\rotatebox{90}{$\Lsh$}
\end{pNiceMatrix}
$}
\\[0.2em]

$\sigma_{U,V}^{\piu}$
&
\Tsymmp{U}{V}
&
\raisebox{-0.4ex}{$
\begin{pNiceMatrix}[first-col,first-row]
\rotatebox{90}{$\Lsh$} & U & V \\
V & 0\cdot \star_{U,V} & 1\cdot \wire{V} \\
U & 1\cdot \wire{U} & 0\cdot \star_{V,U}
\end{pNiceMatrix}
$}
\\[0.2em]

$\t ; \t$
&

    \InputIfFileExists{tapes/seq_comp.tikz}{}{\input{./tikz/tapes/seq_comp.tikz}}

&
matrix multiplication
\\[0.2em]

$\t \piu \t$
&

    \InputIfFileExists{tapes/par_comp.tikz}{}{\input{./tikz/tapes/par_comp.tikz}}

&
direct sum of matrices
\\

\bottomrule
\end{tabular}
\endgroup
\caption{Dictionary of correspondences}\label{fig:dictionary}
\end{figure}

\section{Appendix to Section~\ref{sec:boolcircuits}}\label{app:sec:boolcircuits}

\begin{proposition}\label{prop:foolboolean}
The monoidal functor  $\osemB{-}\colon \Diag{{\SigB}} \to \Sets$ factors as
\[\xymatrix{ \Diag{\SigB} \ar@{->>}[r]& \Diag{\ThB} \ar[r]^{\cong}& \Sets_2 \ar@{^{(}->}[r]& \Sets }\]
where the leftmost and rightmost functor are the obvious quotients and injections and the central arrows is a monoidal isomorphism.
\end{proposition}
\begin{proof}
Faithfulness of the central functor is a direct consequence of Theorem~\ref{thm:completenessbooleancircuits}. Fullness follows from the fact that every function $f\colon 2^n \to 2^m$ corresponds to a vector of $m$ Boolean formulas in $n$ variables, which can be easily encoded as a string diagram in $\Diag{\ThB}$ using Boolean operators, $\CBcopier$ and $\CBdischarger$.
\end{proof}

\begin{proof}[Proof of Lemma~\ref{lemma:multiplexer}] 
  \eqref{ax:M1} and \eqref{ax:M2} can be proved by a simple inductive argument on $m$. \eqref{ax:M3} can be proved by a direct computation with the AND and OR gates. \eqref{ax:N1} and \eqref{ax:N2} follow by induction on the structure of $c$. For the base case, if $c$ is a Boolean gate or one of the structural maps $\id{1},\id{A},\symmt{A}{A}$, then the statement follows from axioms \eqref{ax:C1}, \eqref{ax:C2}, \eqref{ax:C3}  and \eqref{ax:D1}, \eqref{ax:D2}, \eqref{ax:D3} in Figure~\ref{tab:booleanalgebra}. In the case $c;d$, then assuming that the statement holds for $c$ and $d$, we have
  \begin{center}
    
    \InputIfFileExists{tapes/cipriano/booleanscopyable.tikz}{}{\input{./tikz/tapes/cipriano/booleanscopyable.tikz}}

  \end{center}
  \begin{center}
    
    \InputIfFileExists{tapes/cipriano/booleansdiscardable.tikz}{}{\input{./tikz/tapes/cipriano/booleansdiscardable.tikz}}

  \end{center}
  Finally, in the case $c\otimes d$, assuming that the statement holds for $c$ and $d$, we have
  \begin{center}
    
    \InputIfFileExists{tapes/cipriano/booleanscopyable2.tikz}{}{\input{./tikz/tapes/cipriano/booleanscopyable2.tikz}}

  \end{center}
  \begin{center}
    
    \InputIfFileExists{tapes/cipriano/booleansdiscardable2.tikz}{}{\input{./tikz/tapes/cipriano/booleansdiscardable2.tikz}}

   \end{center}
\end{proof}

\section{Appendix to Section~\ref{ssec:partialboolean}}\label{app:sec:probboolcircuits}

\begin{proof}[Proof of Proposition~\ref{prop:soundnesspartialb}]
It is enough to check that the statement holds for the axioms in Figure~\ref{tab:partialbooleanalgebra}. Those in the first two lines are standard. For the third row, simple computations confirm that $\osemPB{-}$ maps the left hand side of \eqref{ax:F6} into $\osemPB{\CBcocopier}$ as defined in \eqref{eq:semcocopier}. For \eqref{ax:F7}, one readily checks that $\osemPB{-}$ maps both the left and the right hand side to the partial function $2\times 2 \to 1$ defined as
\[(x,y)\mapsto \begin{cases}\bullet &\text{if }x=y=1\\ \bot & \text{else}\end{cases}\]
\end{proof}

\begin{proof}[Proof of Lemma~\ref{lemma:failureequalities}]
 \eqref{ax:F9} follows from the equalities  
 \begin{center} 
    \InputIfFileExists{tapes/cipriano/lemmafail1.tikz}{}{\input{./tikz/tapes/cipriano/lemmafail1.tikz}}

  \end{center}
  \begin{center}
  
    \InputIfFileExists{tapes/cipriano/lemmafail2.tikz}{}{\input{./tikz/tapes/cipriano/lemmafail2.tikz}}

  \end{center}
   \eqref{ax:F8} follows from the equalities
  \begin{center}
  
    \InputIfFileExists{tapes/cipriano/lemmafail3.tikz}{}{\input{./tikz/tapes/cipriano/lemmafail3.tikz}}

  \end{center}
  \begin{center}
  
    \InputIfFileExists{tapes/cipriano/lemmafail4.tikz}{}{\input{./tikz/tapes/cipriano/lemmafail4.tikz}}

  \end{center}
  \begin{center}
  
    \InputIfFileExists{tapes/cipriano/lemmafail5.tikz}{}{\input{./tikz/tapes/cipriano/lemmafail5.tikz}}

  \end{center}
  \begin{center}
  
    \InputIfFileExists{tapes/cipriano/lemmafail5mezzo.tikz}{}{\input{./tikz/tapes/cipriano/lemmafail5mezzo.tikz}}

  \end{center}
   \eqref{ax:F10l} follows from the equalities
  \begin{center}
  
    \InputIfFileExists{tapes/cipriano/lemmafail6.tikz}{}{\input{./tikz/tapes/cipriano/lemmafail6.tikz}}

  \end{center}
  \eqref{ax:F10r} can be proved with a similar argument.
\end{proof}

\begin{proof}[Additional details for the proof of Proposition~\ref{prop:decompositionpartialcircuits}]
In the case $c\otimes d$, then assuming that the statement holds for $c$ and $d$, we have 
    \begin{center}
    
    \InputIfFileExists{tapes/cipriano/propdecomposition6.tikz}{}{\input{./tikz/tapes/cipriano/propdecomposition6.tikz}}

    \end{center}
    \begin{center}
    
    \InputIfFileExists{tapes/cipriano/propdecomposition7.tikz}{}{\input{./tikz/tapes/cipriano/propdecomposition7.tikz}}

    \end{center}
    \begin{center}
    
    \InputIfFileExists{tapes/cipriano/propdecomposition8.tikz}{}{\input{./tikz/tapes/cipriano/propdecomposition8.tikz}}

    \end{center}
    now applying the inductive hypothesis for $c$ and $d$, we obtain the statement.\end{proof}

\begin{proof}[Proof of Lemma~\ref{lemma:complete1}]
  We proceed by induction on the structure of $c$. For the base case, if $c\in \SigB \cup \{\id{1},\id{A},\symmt{A}{A}\}$, then $b;D_c=b;\CBdischargern\, \Flip{1}=\Flip{1}$, hence the statement is trivial. If $c$ is $\CBcocopier$, then, since $D_{\scalebox{0.6}{$\CBcocopier$}}$ is the xnor gate, $b$ must be of the form $\Flip{1}\otimes \Flip{0}$ or $\Flip{0}\otimes \Flip{1}$, hence $b;T_{\scalebox{0.6}{$\CBcocopier$}}$ is $\Flip{0}$, by axioms \eqref{ax:B2l} and \eqref{ax:B2r} in Figure~\ref{tab:booleanalgebra}. Now, in the case $c;d$, we have
  \begin{center}
    
    \InputIfFileExists{tapes/cipriano/lemma1fg.tikz}{}{\input{./tikz/tapes/cipriano/lemma1fg.tikz}}

  \end{center}
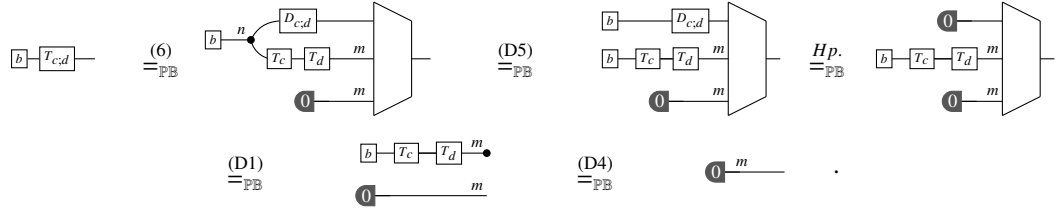
  \begin{center}
    
    \InputIfFileExists{tapes/cipriano/lemma1fgpt2.tikz}{}{\input{./tikz/tapes/cipriano/lemma1fgpt2.tikz}}

  \end{center}
   Finally, in the case $c\otimes d$, since  $b\eqPB b_1\otimes b_2$ for some $b_1\in \Diag{\SigB}[1,A^{n'}]$ and $b_2\in \Diag{\SigB}[1,A^{n''}]$ such that $n'+n''=n$, we have 
  \begin{center}
    
    \InputIfFileExists{tapes/cipriano/lemma1fperg.tikz}{}{\input{./tikz/tapes/cipriano/lemma1fperg.tikz}}

  \end{center}
  \begin{center}
    
    \InputIfFileExists{tapes/cipriano/lemma1fpergpt2.tikz}{}{\input{./tikz/tapes/cipriano/lemma1fpergpt2.tikz}}

  \end{center}
\end{proof}

\begin{lemma}\label{lemma:complete2}
  Let $b\in \Diag{\SigB}[1,A^n]$ and $c\in \Diag{PB}[A^n,A^m]$.
  \begin{itemize}
    \item[i.] If $\osemPB{b;c}=\bot$ , then $\osemPB{b;D_c} = 0$;
    \item[ii.] If $\osemPB{b;c}=\osemPB{b'}$ for some $b'\in \Diag{\SigB}[1,A^m]$, then $\osemPB{b;D_c}=1$ and $\osemPB{b;T_c}=\osemPB{b'}$.
  \end{itemize}
\end{lemma}
\begin{proof}
Since $D_c \in \Diag{\SigB}[A^n,1]$, then $\osemPB{b;D_c} = \osemB{b;D_c}$, i.e., it is a (total) function of type $1 \to 2$. Hence $\osemPB{b;D_c}$ is either $0$ or $1$. 

Suppose that $\osem{b;D_c}=1$, or, equivalently, $b;D_c\eqPB\Flip{1}$ thanks to Theorem~\ref{thm:completenessbooleancircuits}. Then
  \begin{align}
    \osemPB{b;c}&=  \osemPB{b;\ncopier;((D_c;\coflip{1})\otimes T_c)} \tag{Prop. \ref{prop:decompositionpartialcircuits}}\\
    &= \osemPB{(b\otimes b);((D_c;\coflip{1})\otimes T_c)} \tag{\eqref{ax:N2}}\\
    &= \osemPB{(b;D_c;\coflip{1})\otimes (b;T_c)} \tag{SMC}\\
    &= \osemPB{(\Flip{1};\coflip{1})\otimes (b;T_c)} \tag{$b;D_c\eqPB\Flip{1}$}\\
    &= \osemPB{b;T_c}. \tag{\eqref{ax:F10l}}
\end{align}
Instead, suppose that $\osem{b;D_c}=0$, or, equivalently, $b;D_c\eqPB\Flip{0}$ thanks to Theorem~\ref{thm:completenessbooleancircuits}. Then, 
\begin{align}
    \osemPB{b;c}  &= \osemPB{b;\ncopier;((D_c;\coflip{1})\otimes T_c)} \tag{Prop. \ref{prop:decompositionpartialcircuits}}\\
    &= \osemPB{(b\otimes b);((D_c;\coflip{1})\otimes T_c)} \tag{\eqref{ax:N2}}\\
    &= \osemPB{(b;D_c;\coflip{1})\otimes (b;T_c)} \tag{SMC}\\
    &= \osemPB{(\Flip{0};\coflip{1})\otimes (b;T_c)} \tag{$b;D_c\eqPB\Flip{0}$}\\
    &= \bot \notag
\end{align}

To prove i.,  we assume that $\osemPB{b;c}=\bot$ and $\osem{b;D_c}=1$ and  immediately obtain a contradiction: by the first derivation above $\osemPB{b;T_c}=\osemPB{b;c}=\bot$, but this is not possible since $b;T_c \in \Diag{\SigB}[1,A^m]$, i.e., it denotes a total Boolean function. Hence, if  $\osemPB{b;c}=\bot$, then $\osem{b;D_c}$ should be $0$.

Now, to prove ii., we first assume that $\osemPB{b;c}=\osemPB{b'}$ and $\osem{b;D_c}=1$ and immediately obtain a contradiction: by the second derivation above $\osemPB{b;c}=\bot$. Hence, if  $\osemPB{b;c}=\osemPB{b'}$, then $\osem{b;D_c}$ should be $1$. By the first derivation $\osemPB{b;T_c}=\osemPB{b;c}=\osemPB{b'}$.

\end{proof}

\begin{proof}[Proof of Lemma~\ref{lemma:complete3}]
There are two cases: either $\osemPB{b;c}=\osemPB{b;d}=\bot$, or $\osemPB{b;c}=\osemPB{b;d}\neq \bot$. In the first case, by Lemma~\ref{lemma:complete2} (i), we have $\osemPB{b;D_c}=0=\osemPB{b;D_d}$, and by Theorem~\ref{thm:completenessbooleancircuits} we have $b;D_c\eqPB b;D_d\eqPB\Flip{0}$. Hence, Lemma~\ref{lemma:complete1} implies ${b;T_c}\eqPB\Flipm{0}={b;T_d}$. In the second case, by Lemma~\ref{lemma:complete2} (ii), we have $\osemPB{b;D_c}=1=\osemPB{b;D_d}$ and $\osemPB{b;T_c}=\osemPB{b'}=\osemPB{b;T_d}$ for some $b'\in \Diag{\SigB}[1,A^m]$.
  \end{proof}

  \begin{proof}[Proof of Proposition~\ref{prop:isopartialboolean}]
Recall that $\osemPB{-}\colon \Diag{{\SigPB}} \to \Cat{Par}$ is defined on objects as $\osemPB{A^n}=2^n$. Hence for all $c\in \Diag{{\SigPB}}[A^n,A^m]$, $\osemPB{c}\in \Cat{Par}[2^n,2^m]$, namely 
 $\osemPB{c}$ is an arrow of $\Cat{Par}_2$. Thus  $\osemPB{-}\colon \Diag{{\SigPB}} \to \Cat{Par}$ factors as
 \[\xymatrix{ \Diag{\SigPB} \ar[r] & \Cat{Par}_2 \ar@{^{(}->}[r]& \Cat{Par} }\text{.}\]
 Moreover, by soundness (Proposition~\ref{prop:soundnesspartialb}),  it factors through $\Diag{{\ThPB}}$: 
\[\xymatrix{ \Diag{\SigPB} \ar@{->>}[r]& \Diag{\ThPB} \ar[r]& \Cat{Par}_2 \ar@{^{(}->}[r]& \Cat{Par} }\]
The central functor is faithful by completeness. To conclude that it is an isomorphism, it is enough to prove fullness, i.e., that for any partial function $f\colon 2^n\to2^m$, there exists some diagram $d\colon A^n\to A^m$ such that $\osemPB{d}=f$. But this is trivial by first observing that for any such partial function, there exist total functions $D_f\colon2^n \to 2 $ and $T_f\colon 2^n \to 2^m$ that decompose $f$ as in Proposition~\ref{prop:decompositionpartialcircuits} and then make use of Proposition~\ref{prop:foolboolean}.
\end{proof}

\begin{proof}[Proof of Lemma~\ref{lemma:caratterizzazionecircuiti0n}]
Recall that $\osemPB{c}$ is a partial function of type $1 \to 2^n$. Thus, it is either $\bot$ or  of the form $\osemPB{b}$ for some $b\in \Diag{\SigB}[1,A^n]$. In the first case, simple computations confirm that  $\osemPB{\Flipn{\bot}}=\bot$, thus $\osemPB{c}=\osemPB{\Flipn{\bot}}$. By Theorem~\ref{thm:completenesspartialcircuits}, $c\eqPB \Flipn{\bot}$. In the second case, $\osemPB{c}=\osemPB{b}$ and, again by Theorem~\ref{thm:completenesspartialcircuits}, $c\eqPB b$.
\end{proof}
 
\section{Appendix to Section \ref{sec:encodingprobboolcircuits}}\label{app:sec:encodingprobboolcircuits}
\begin{proof}[Proof of Proposition \ref{prop:encoding}] 
By induction on $c$. The base case $c=\Flip{p}$ is proved before the statement of the proposition. For the base case $c\in \SigPB\cup\{\id{A},\id{1},\symmt{A}{A}\}$,  $\CBdsem{\encoding{c}} = \CBdsem{\tapeFunct{c} }$ by the definition in Table \ref{tab:encoding} and $\CBdsem{\tapeFunct{c} } = \ParJ(\osemPB{c})$ by \eqref{eq:semtapes}. The latter is exactly $\osem{c}$.
The inductive cases follow immediately from the fact that $\CBdsem{-}$ is a morphism of rig categories and hence preserve $\per$.
\end{proof}
\section{Appendix to Section \ref{sec:preliminaries}}\label{app:sec:cbc}

\begin{proof}[Proof of Theorem~\ref{thm:TCconvexbiproductcategory}]
    See \cite[Theorem 3]{probbooltapesfossacs}.
\end{proof}

\begin{proof}[Proof of Theorem~\ref{thm:syntacticadjunction}]
    See \cite[Theorem 4]{probbooltapesfossacs}.
\end{proof}
\section{Appendix of Section \ref{sec:completenessprobbooltapes}}\label{app:{sec:completenessprobbooltapes-NEW2}}

\begin{proof}[Proof of Proposition \ref{prop:soundnesstapes}]
Since the axioms in $\ThPB$ are sounds with respect to $\osem{-}$ and since $\osem{-}=\dsem{\encoding{-}}$, these axioms are also sound for $\dsem{-}$.
For \eqref{ax:AXPBP1} and \eqref{ax:AXPBP2}, it enough to check that the left and the right hand sides are mapped into the same $\KlD$-arrows by $\dsem{-}$.
For \eqref{ax:canc}, it is enough to observe that cancellativity holds in $\KlD$.
\end{proof}

Consider the congruence $\tapeeqPB$ on $\CatT{\Diag{\SigPB}}$ generated by the rules in \eqref{eq:congr2}
\begin{equation}\label{eq:congr2}
        \scalebox{0.85}{$
        \centering
        \begin{array}{c}
        \begin{array}{c@{\qquad\qquad}c@{\qquad\qquad}c@{\qquad\qquad}c}
            \inferrule*[right=($\ThPB$)]{c =_{\ThPB} d}{\tape{c} \tapeeqPB \tape{d}}
            &
            \inferrule*[right=($\textsc{R}$)]{-}{\t \tapeeqPB \t}
            &
            \inferrule*[right=($\textsc{S}$)]{\t_1 \tapeeqPB \t_2}{\t_2 \tapeeqPB \t_1}
            &    
            \inferrule*[right=($\textsc{T}$)]{\t_1 \tapeeqPB \t_2 \quad \t_2 \tapeeqPB \t_3}{\t_1 \tapeeqPB \t_3}
        \end{array}
        \\[10pt]
        \begin{array}{c@{\qquad}c@{\qquad}c}
            \inferrule*[right=($;$)]{\t_1 \tapeeqPB \t_2 \quad \s_1 \tapeeqPB \s_2}{\t_1;\s_1 \tapeeqPB \t_2;\s_2}
            &
            \inferrule*[right=($\piu$)]{\t_1 \tapeeqPB \t_2 \quad \s_1 \tapeeqPB \s_2}{\t_1\piu\s_1 \tapeeqPB \t_2 \piu \s_2}
            &
            \inferrule*[right=($\per $)]{\t_1 \tapeeqPB \t_2 \quad \s_1 \tapeeqPB \s_2}{\t_1\per \s_1 \tapeeqPB \t_2 \per \s_2}       
        \end{array}
        \end{array}
        $}
\end{equation}

\begin{proposition}\label{prop:quotienttheory}
For all $\s , \t \in \CatT{\Diag{\SigPB}}$, $\s \tapeeqPB \t$ iff $\CatT{Q_{\ThPB}}(\s) = \CatT{Q_{\ThPB}}(\t)$.
\end{proposition}
\begin{proof}
  Let $\CatT{\Diag{\SigPB}}_{\tapeeqPB}$ be the quotient of $\CatT{\Diag{\SigPB}}$ by $\tapeeqPB$: objects are those of  $\CatT{\Diag{\SigPB}}$, arrows are $\tapeeqPB$-equivalence classes of arrows of $\CatT{\Diag{\SigPB}}$, hereafter denoted by $[\t]_{\tapeeqPB}$. We prove that $\CatT{\Diag{\SigPB}}_{\tapeeqPB}$ is isomorphic to $\CatT{\Diag{\ThPB}}$.

We first define the functor $F \colon \CatT{\Diag{\SigPB}}_{\tapeeqPB} \to \CatT{\Diag{\ThPB}}$ inductively as 
\[\begin{array}{ccccc}
F(\diagp{U})\defeq \; \diagp{U} &  F(\,\bangp{U})\defeq \bangp{U} & F(\tapeFunct{c})\defeq \tapeFunct{[c]_{\ThPB}} &  F(\cobang{U})\defeq \cobang{U} &
F(\codiag{U})\defeq\codiag{U}  \\
 F(\id{U})\defeq \id{U} & F(\id{\zero})\defeq\id{\zero} & F(\sigma_{U,V}^{\piu})\defeq \sigma_{U,V}^{\piu}& F(\t_1 ; \t_2)\defeq F(\t_1);F(\t_2) & F(\t_1 \piu \t_2) \defeq F(\t_1)\piu F(\t_2)
\end{array}\]
Since arrows in $\CatT{\Diag{\SigPB}}_{\tapeeqPB}$ are $\tapeeqPB$-equivalence classes, we need to prove that the functor is well-defined, namely that if $\s \tapeeqPB  \t$ then $F(\s)=F(\t)$. This is done by induction on the rules in \eqref{eq:congr2}. For the case of the rule ($\ThPB$), observe that $\s=\tape{c}$, $\t=\tape{d}$ and $c=_{\ThPB}d$. By the definition of $F$, $F(\s)=F(\t)$. All the other cases are trivial, with the only exception of the rule $(\otimes)$ where one observes that, since  $=_{\ThPB}$ is a congruence w.r.t. $\otimes$, then $F$ is a morphism of rig categories.

The functor  $G \colon  \CatT{\Diag{\ThPB}} \to \CatT{\Diag{\SigPB}}_{\tapeeqPB} $ is defined on $\tapeFunct{[c]_{\ThPB}}$ as $G(\tapeFunct{[c]_{\ThPB}})\defeq [\tapeFunct{c}]_{\tapeeqPB}$ and then inductively in the same style of the functor $F$ above. Observe that $F(G(\tapeFunct{[c]_{\ThPB}}))=\tapeFunct{[c]_{\ThPB}}$ and $G(F([\tapeFunct{c}]_{\tapeeqPB})) = [\tapeFunct{c}]_{\tapeeqPB}$. A simple inductive argument confirms that $F$ and $G$ are inverses of each other.
\end{proof}

\begin{proof}[Proof of Proposition \ref{prop:twoequivalence}]
 
  We first prove by induction on the rules in \eqref{eq:congr3} that if $\s \dot{\sim} \t$ then $\CatT{Q_{\ThPB}}(\s) \eqsyn \CatT{Q_{\ThPB}}(\t)$.
Consider the case of the rule $(\ThPB)$: in this case $\s=\tapeFunct{c}$, $\t=\tapeFunct{d}$ and $c\eqPB d$. Thus $\CatT{Q_{\ThPB}}(\s)= \CatT{Q_{\ThPB}}(\t)$.
All the other cases trivially follow from the fact that $\CatT{Q_{\ThPB}}$ is a morphism of rig categories: see \cite{bonchi2025tapediagramsmonoidalmonads}.

Again, by induction on the rules, we prove that if $\CatT{Q_{\ThPB}}(\s) \eqsyn \CatT{Q_{\ThPB}}(\t)$ then $\s \dot{\sim} \t$.
Here the only non-trivial cases are $(R)$, $\eqref{ax:AXPBP1}$ and $\eqref{ax:AXPBP2}$. For $(R)$, if $\CatT{Q_{\ThPB}}(\s) = \CatT{Q_{\ThPB}}(\t)$, then by Proposition \ref{prop:quotienttheory}, $\s \tapeeqPB \t$. Hence, $\s \dot{\sim} \t$. For $\eqref{ax:AXPBP1}$, if $\CatT{Q_{\ThPB}}(\s) \overset{\eqref{ax:AXPBP1}}{=} \CatT{Q_{\ThPB}}(\t)$, then $\s$ is $\tapeeqPB$-equivalent to the left hand side of \eqref{ax:AXPBP1} and $\t$ is $\tapeeqPB$-equivalent to the right hand side of \eqref{ax:AXPBP1}. Hence, they are $\dot{\sim}$-equivalent, and applying transitivity of $\dot{\sim}$, we conclude that $\s \dot{\sim} \t$. The case of $\eqref{ax:AXPBP2}$ is analogous.
\end{proof}

\begin{proof}[Proof of Lemma \ref{lemma:bottomstar2}]
 \begin{align}
  \Flipn{\bot}[t]&=\Flip{\bot}[t] ;\ncopierbis[t]  \tag{Def.\ of $\Flipn{\bot}$}\\
  &\sim \Tcounitempty{};\Tunit{}; \ncopierbis[t]  \tag{Axiom \ref{ax:AXPBP2}}\\
  &=  \Tcounitempty{};\Tunitm{n} \tag{\eqref{ax:tapes:cobangnat}}\\
  &= \star_{A^0,A^n}. \tag{Def. of $\star$}
 \end{align}
\end{proof}

\begin{proof}[proof of Lemma~\ref{lemma:axiomsninput}]
By induction  on $n$. The base case is trivial. For the inductive step, we have
 \begin{align}
  \frac{1}{2^{n+1}}\cdot \id{A^{n+1}}&= \frac{1}{2^n}\frac{1}{2}\cdot( \id{A}\otimes \id{A^n}) \notag\\
  &= \frac{1}{2^n}\cdot\left(\frac{1}{2}\cdot \id{A}\otimes \id{A^n} \right) \tag{\eqref{eq:monoidalenrichment}}\\
  &= (\frac{1}{2}\cdot \id{A})\otimes (\frac{1}{2^n}\cdot \id{A^n}) \tag{\eqref{eq:monoidalenrichment}}\\
  &\sim (\frac{1}{2}\cdot \id{A})\otimes \left(\sum_{\overrightarrow{b}\in \Cat{B}[1,A^n]}^{} \frac{1}{2^n}\cdot\overleftarrow{b};\overrightarrow{b}\right) \tag{Inductive hypothesis}\\
  &\sim (\coflip{1}[t];\Flip{1}[t] +_{\frac{1}{2}} \coflip{0}[t];\Flip{0}[t]) \otimes \left(\sum_{\overrightarrow{b}\in \Cat{B}[1,A^n]}^{} \frac{1}{2^n}\cdot\overleftarrow{b};\overrightarrow{b}\right) \tag{Axiom \ref{ax:AXPBP1}}\\
  &= \sum_{\overrightarrow{b}\in \Cat{B}[1,A^n]}^{} \frac{1}{2^n}\cdot\left((\coflip{1}[t];\Flip{1}[t] +_{\frac{1}{2}} \coflip{0}[t];\Flip{0}[t]) \otimes (\overleftarrow{b};\overrightarrow{b})\right) \tag{\eqref{eq:monoidalenrichment}}\\
  &=\sum_{\overrightarrow{b}\in \Cat{B}[1,A^n]}^{} \frac{1}{2^n}\cdot\left(((\coflip{1}[t];\Flip{1}[t]) \otimes (\overleftarrow{b};\overrightarrow{b})) +_{\frac{1}{2}} ((\coflip{0}[t];\Flip{0}[t]) \otimes (\overleftarrow{b};\overrightarrow{b}))\right) \tag{SMC}\\
  &= \sum_{\overrightarrow{b}\in \Cat{B}[1,A^n]}^{} \frac{1}{2^n}\cdot\left((\coflip{1}[t]\otimes \overleftarrow{b});(\Flip{1}[t]\otimes \overrightarrow{b}) +_{\frac{1}{2}} (\coflip{0}[t]\otimes \overleftarrow{b});(\Flip{0}[t]\otimes \overrightarrow{b})\right) \tag{\eqref{eq:monoidalenrichment}}
 \end{align}
 Now, since the set $\{\Flip{1}[t]\otimes \overrightarrow{b}, \Flip{0}[t]\otimes \overrightarrow{b}\}$ as $\overrightarrow{b}\in \Cat{B}[1,A^n]$ is exactly $\Cat{B}[1,A^{n+1}]$, then, rearranging the above sum we obtain $\frac{1}{2^{n+1}}\cdot \id{A^{n+1}}= \sum_{\overrightarrow{b}\in \Cat{B}[1,A^{n+1}]}^{} \frac{1}{2^{n+1}}\cdot\overleftarrow{b};\overrightarrow{b}$, as desired.
\end{proof}

\begin{proof}[Proof of Lemma \ref{lemma:PBP2}]
Consider the following derivation: 
\begin{align}
  \frac{1}{2^n} \cdot \t &= \frac{1}{2^n}\cdot (\id{A^n};\t) \notag\\
  &= (\frac{1}{2^n}\cdot \id{A^n});\t \tag{$\Cat{PCA}$-enrichment} \\ 
  &\sim (\sum_{\overrightarrow{b}\in \Cat{B}[1,A^n]}^{} \frac{1}{2^n}\cdot\overleftarrow{b};\overrightarrow{b});\t \tag{Lemma \ref{lemma:axiomsninput}}\\
  &= \sum_{\overrightarrow{b}\in \Cat{B}[1,A^n]}^{} \frac{1}{2^n}\cdot\overleftarrow{b};\overrightarrow{b};\t \tag{\eqref{eq:enr}}\\
  &\sim \sum_{\overrightarrow{b}\in \Cat{B}[1,A^n]}^{} \frac{1}{2^n}\cdot\overleftarrow{b};\overrightarrow{b};\s \tag{Hp.}\\
  &=(\sum_{\overrightarrow{b}\in \Cat{B}[1,A^n]}^{} \frac{1}{2^n}\cdot\overleftarrow{b};\overrightarrow{b});\s \tag{\eqref{eq:enr}}\\
  &\sim (\frac{1}{2^n}\cdot \id{A^n});\s \tag{Lemma \ref{lemma:axiomsninput}}\\
  &= \frac{1}{2^n} \cdot \s. \tag{$\Cat{PCA}$-enrichment} 
\end{align}
Hence, applying axiom \eqref{ax:canc}, we conclude that $\t \eqsyn \s$.
\end{proof}

Denote with $G^\flat: \CatT{\Cat{C}}\to \Cat{D}$ the arrow obtained by a functor $G\colon \Cat{C}\to U(\Cat{D})$ through the adjunction $\CatT{-}\dashv U$ in Theorem~\ref{thm:syntacticadjunction}. As proven in \cite[Theorem 4]{probbooltapesfossacs} it is given by the following inductive definition on the structure of arrows in $\CatT{\Cat{C}}$:
\begin{equation}\label{eq:bemolle}
\begin{aligned}
G^\flat(\id{U}) &\defeq \id{G(U)} 
&\qquad G^\flat(\tapeFunct{c}) &\defeq G(c) \\
G^\flat(\diagp{P}) &\defeq\, \diagp{G(P)}  
&\qquad G^\flat(\bang{P}) &\defeq \bang{G(P)} \\
G^\flat(\codiag{P}) &\defeq \codiag{G(P)} 
&\qquad G^\flat(\cobang{P}) &\defeq \cobang{G(P)} \\
G^\flat(\t;\s) &\defeq\, G^\flat(\t);G^\flat(\s)
&\qquad G^\flat(\t \oplus \s) &\defeq\, G^\flat(\t)\oplus G^\flat(\s) \\
G^\flat(\symmp{P}{Q}) &\defeq\, \symmp{G(P)}{G(Q)}
&\qquad G^\flat(\id{\zero}) &\defeq \id{G(\zero)}
\end{aligned}
\end{equation}

Denote with $\jcompiso$ the composition of the isomorphism in Proposition~\ref{prop:isopartialboolean} and the functor $\ParJ_2$:
\[\xymatrix{  \Diag{\mathbb{\SigPB }} \ar[r]^{\cong}& \Cat{Par}_2 \ar@{->}[r]^{\ParJ_2}& \KlD }\]

\begin{lemma}\label{lemma:J_2}
$\comp=\jcompiso^\flat$.
\end{lemma}
\begin{proof}
  It follows from \eqref{eq:adjunctions}, Theorem~\ref{thm:syntacticadjunction}  and Proposition~\ref{prop:isopartialboolean}.
\end{proof}

\begin{proof}[Proof of Lemma \ref{lemma:diagramma-commutativo}]
    In order to prove the statement we exploit Proposition~\ref{prop:quotienttheory} which provides an isomorphism $F\colon\CatT{\Diag{\SigPB}}_{\tapeeqPB} \to \CatT{\Diag{\ThPB}}$, such that $Q_{\tapeeqPB};F= {\CatT{Q_\ThPB}}$, where $Q_{\tapeeqPB}\colon \CatT{\Diag{\SigPB}} \to \CatT{\Diag{\SigPB}}_{\tapeeqPB}$ is the quotient functor.
     Then, since $\dsem{-}$ is clearly sound with respect to ${\tapeeqPB}$, it is enough to prove that $\comp\circ F\circ Q_{\tapeeqPB} = \dsem{-}$. Since $\dsem{-}$ is defined inductively, it is enough to check that $\comp\circ F\circ Q_{\tapeeqPB}$ and $\dsem{-}$ coincide on generators.
    For instance consider the generator $\diagp{A^n}$:
    \begin{align*}
      \comp\circ F\circ Q_{\tapeeqPB}(\diagp{A^n})&= \comp\circ F(\diagp{A^n}) \tag{Def. $Q_{\tapeeqPB}$}\\
      &= \comp(\diagp{A^n}) \tag{Def. $F$}\\
      &=\jcompiso^\flat(\diagp{A^n}) \tag{Lemma \ref{lemma:J_2}}\\
      &=\, \diagp{2^n} \tag{\ref{eq:bemolle}}\\
      &=\dsem{\diagp{A^n}}. \tag{Def. $\dsem{-}$}
    \end{align*}
    The cases for the generators \,$\bangp{A^n}$, $\cobang{A^n}$ and $\codiag{A^n}$ are analogous. For the generator $\tapeFunct{c}$, we have
    \begin{align}
      \comp\circ F\circ Q_{\tapeeqPB}(\tapeFunct{c})&= \comp\circ F([\tapeFunct{c}]_{\tapeeqPB}) \tag{Def. $Q_{\tapeeqPB}$}\\
      &= \comp(\tapeFunct{[c]_{\ThPB}}) \tag{Def. $F$}\\
      &=\jcompiso^\flat(\tapeFunct{[c]_{\ThPB}}) \tag{Lemma \ref{lemma:J_2}}\\
      &= \jcompiso([c]_{\ThPB}) \tag{\ref{eq:bemolle}}\\
      &=\ParJ_2(\osemPB{c}) \tag{Proposition \ref{prop:isopartialboolean}}\\
      &= \dsem{\tapeFunct{c}}. \tag{Def. $\dsem{-}$}
\end{align}
The remaining cases are trivial.
Hence, by the universal property of
$\CatT{\Diag{\SigPB}}_{\tapeeqPB}$ we obtain that $\comp\circ F= I\circ \eqsynq{Q}\circ F$, and since $F$ is an isomorphism, we obtain the statement.

    \end{proof}

\begin{proof}[Proof of Lemma \ref{lemma:Ffaith}]

  Denote with $\ParJ_2^\sharp\colon \Cat{Par}_2^+\to \KlD$ the arrow obtained by $\ParJ_2$ through the adjunction $(-)^+\dashv U$ in \eqref{eq:adjunctions}. By construction, it is the identity-on-objects functor sending a subdistribution $d\in \Cat{Par}_2^+[1,2^m]$ into the subdistribution $\ParJ_2^\sharp(d)= \sum_{f\in \Cat{Par}_2[1,2^m]} d(f)\cdot{\ParJ_2(f)}$. Hence, if $d\in \Sets_2^+[1,2^m]$, then $\ParJ_2^\sharp(d)= \sum_{f\in \Sets_2[1,2^m]} d(f)\cdot{\ParJ_2(f)}=\sum_{f\in \Sets_2[1,2^m]} d(f)\cdot{\delta_{f(\bullet)}}$, where $f(\bullet)\in 2^m$. Then, $\ParJ_2^\sharp$ restricted to $\Sets_2^+[1,2^m]$ provides the obvious bijection $$\Dis(\Sets_2[1,2^m])\cong \Dis(2^m) \cong   \KlD_2[1,2^m]\text{.}$$ 

The statement now follows from the fact that by construction $\eta_{\Cat{Par}_2^+};\ParJ_2'=\ParJ_2^\sharp$, where $\eta$ is the unit of the adjunction $\stmat{-}\dashv U$ in \eqref{eq:adjunctions}, and that fact that the isomorphisms 
\[\CatT{\Diag{\ThPB}} \cong \stmat{\Diag{\ThPB}^+}\cong \stmat{\Cat{Par}_2^+}\] 
restricted to $[A^0,A^m]$ factor through the composition
\[\CatT{\Diag{\ThPB}}[A^0,A^m] \cong \Diag{\ThPB}^+[A^0,A^m]\cong \Cat{Par}_2^+[2^0,2^m]\overset{\eta_{\Cat{Par}_2^+}}{\to }\stmat{\Cat{Par}_2^+}[2^0,2^m],\]
where the first isomorphism is given by Corollary~\ref{cor:isotapematrices} and the obvious isomorphism $\stmat{\Diag{\ThPB}^+}[A^0,A^m]\cong \Diag{\ThPB}^+[A^0,A^m]$, and the second isomorphism is induced by Proposition~\ref{prop:isopartialboolean}.
\end{proof}

\begin{proof}[Proof of Lemma~\ref{lemma:PBP0}]
Since $\CatT{\Diag{\ThPB}}\cong \stmat{\Diag{\ThPB}^+}$,  $\t$ corresponds to  a subdistribution on $\Diag{\ThPB}[A^0,A^m]$. By Lemma~\ref{lemma:caratterizzazionecircuiti0n}, each partial Boolean circuit in $\Diag{\ThPB}[A^0,A^m]$ is either a Boolean circuit in $\Diag{\SigB}[1,A^m]$ or it is of the form $\Flipm{\bot}$. Hence, we can rearrange the subdistribution corresponding to $\t$ into a tape of the desired form.
\end{proof}

\begin{proof}[Proof of Lemma \ref{lemma:PBP1}]
  By Lemma~\ref{lemma:PBP0}, we can write 
  \begin{equation}\label{eq:BLABLA}
  \t=(\sum_{i=1}^{n}p_i\cdot \overrightarrow{b}_i)+_p(\star_{A^0,A^m} +_q \Flipm{\bot}[t]) \text{ and }\t'=(\sum_{j=1}^{n'}p'_j\cdot \overrightarrow{b}'_j)+_{p'}(\star_{A^0,A^m} +_{q'} \Flipm{\bot}[t])\end{equation} 
  for some $p_i,p'_j\in (0,1)$, for $i=1,\dots,n$ and $j=1,\dots,n'$, $p,p',q,q'\in [0,1]$  such that $\sum_{i=1}^{n}p_i= 1$ and $\sum_{j=1}^{n'}p'_j= 1$, where $\overrightarrow{b}_i\colon A^0 \to A^m\in \Cat{B}[1,A^m]$ and $\overrightarrow{b}'_j\colon A^0 \to A^m\in \Cat{B}[1,A^m]$. Then,
  \begin{align}
    \comp(\t)&= \comp((\sum_{i=1}^{n}p_i\cdot \overrightarrow{b}_i) +_p(\star_{A^0,A^m} +_q \Flipm{\bot}[t])) \tag{\ref{eq:BLABLA}}\\
    &= \comp((\sum_{i=1}^{n}p_i\cdot \overrightarrow{b}_i)) +_p(\comp(\star_{A^0,A^m}) +_q \comp(\Flipm{\bot}[t])) \tag{$\comp$ pca-enriched}\\
    &= \comp((\sum_{i=1}^{n}p_i\cdot \overrightarrow{b}_i)) +_p (\star_{1,2^m} +_q \star_{1,2^m}) \tag{$\comp(\star_{A^0,A^m})=\star_{1,2^m}=\comp(\Flipm{\bot}[t])$}\\
    &= p\cdot \comp(\sum_{i=1}^{n}p_i\cdot \overrightarrow{b}_i) \tag{Idemp.\ in (\ref{eq:pca})}\\
    &=\comp(p\cdot \sum_{i=1}^{n}p_i\cdot \overrightarrow{b}_i)\tag{$\comp$ pca-enriched}
  \end{align}
  and similarly $\comp(\t')=\comp(p'\cdot \sum_{j=1}^{n'}p'_j\cdot \overrightarrow{b}'_j)$. Hence, $\comp(p\cdot \sum_{i=1}^{n}p_i\cdot \overrightarrow{b}_i)=\comp(p'\cdot \sum_{j=1}^{n'}p'_j\cdot \overrightarrow{b}'_j)$, and Lemma~\ref{lemma:Ffaith} implies that  
  \begin{equation}\label{BLA2}p\cdot \sum_{i=1}^{n}p_i\cdot \overrightarrow{b}_i=p'\cdot \sum_{j=1}^{n'}p'_j\cdot \overrightarrow{b}'_j\text{.}\end{equation} Thus,
  \begin{align}
    \t&=(\sum_{i=1}^{n}p_i\cdot \overrightarrow{b}_i) +_p(\star_{A^0,A^m} +_q \Flipm{\bot}[t]) \tag{\eqref{eq:BLABLA}}\\
    &\sim  \sum_{i=1}^{n}p_i\cdot \overrightarrow{b}_i +_p(\star_{A^0,A^m} +_q \star_{A^0,A^m}) \tag{Lemma \ref{lemma:bottomstar2}}\\
    &= \sum_{i=1}^{n}p_i\cdot \overrightarrow{b}_i +_p \star_{A^0,A^m} \tag{Idemp.\ in (\ref{eq:pca})}\\
    &= p\cdot \sum_{i=1}^{n}p_i\cdot \overrightarrow{b}_i\tag{Def. $p\cdot-$}\\
    &= p'\cdot \sum_{j=1}^{n'}p'_j\cdot \overrightarrow{b}'_j \tag{\eqref{BLA2}}\\
    &= \sum_{j=1}^{n'}p'_j\cdot \overrightarrow{b}'_j +_{p'} \star_{A^0,A^m} \tag{Def. $p\cdot-$}\\
    &\sim \sum_{j=1}^{n'}p'_j\cdot \overrightarrow{b}'_j +_{p'}(\star_{A^0,A^m} +_{q'} \star_{A^0,A^m}) \tag{Idemp.\ in (\ref{eq:pca})}\\
    &= \sum_{j=1}^{n'}p'_j\cdot \overrightarrow{b}'_j +_{p'}(\star_{A^0,A^m} +_{q'} \Flipm{\bot}[t]) \tag{Lemma \ref{lemma:bottomstar2}}\\
    &= \t'. \tag{\eqref{eq:BLABLA}}
  \end{align}
\end{proof}

\begin{proof}[Additional details for the proof of Theorem~\ref{thm:completenessPBPtapes}]
    In the main text we have already proved that $I$ is faithful for arrows of type $A^n \to A^m$.
     This fact and convex products easily entail that $I$ is faithful for arrows $\s,\t \colon A^k \to \bigoplus_{i=1}^nA^{m_i}$. 
Indeed, by Theorem~\ref{thm:TCconvexbiproductcategory}, 
\begin{equation}\label{eq:boh}\s=\langle \s_1, \dots, \s_n \rangle_{\vec{q}} \quad \t=\langle \t_1, \dots, \t_n \rangle_{\vec{p}}\end{equation} 
for some $\vec{q}=q_1,\dots, q_n$, $\vec{p}=p_1,\dots, p_n$, $\s_i \colon A^k \to A^{m_i}$ and $\t_i \colon A^k \to A^{m_i}$. Thus
\begin{align*}
I([\s]_\eqsyn) = I([\t]_\eqsyn) & \Rightarrow \comp(\s)=\comp(\t) \tag{Lemma \ref{lemma:diagramma-commutativo}}\\
& \Rightarrow \text{for all }j \text{, } \comp(\s);\comp(\pi_j)=\comp(\t);\comp(\pi_j) \\
& \Rightarrow \text{for all }j\text{, } \comp(\s;\pi_j)=\comp(\t;\pi_j) \tag{Functoriality} \\
& \Rightarrow  \text{for all }j\text{, }  \s;\pi_j\sim\t;\pi_j \tag{Previous implication}\\
& \Rightarrow  \text{for all }j\text{, }  q_j \cdot \s_j\sim p_j \cdot \t_j \tag{\ref{eq:boh}}\\
& \Rightarrow  \text{for all }j\text{, }  [q_j \cdot \s_j]_\sim = [p_j \cdot \t_j]_\sim 
\end{align*}
Now \cite[Proposition 6]{probbooltapesfossacs} states that $\CatT{\Diag{\ThPB}}_\sim$ is a convex biproduct category and that the quotient functor $Q_\eqsyn\colon \CatT{\Diag{\ThPB}} \to \CatT{\Diag{\ThPB}}_\sim$ is a morphism of convex biproduct categories.  Hence, \cite[Corollary 4]{arxivprobbooltapes} implies that $(\bigoplus_{j=1}^nA^{m_j}, [\pi_j]_\sim)$ is the $n$-ary convex product of $A^{m_1}, \dots, A^{m_n}$ in $\CatT{\Diag{\ThPB}}_\sim$ and $ [\s]_\sim$ is the unique arrow such that for all $j=1,\dots,n$,
\begin{align}
 [\s]_\sim;[\pi_j]_\sim &= [\s;\pi_j]_\sim \tag{Functoriality of $Q_\eqsyn$}\\
  &=[q_j \cdot \s_j]_\sim \tag{\ref{eq:boh}}\\
 &=q_j\cdot [\s_j]_\sim; \tag{pca-enrichment of $\CatT{\Diag{\ThPB}}_\sim$}
\end{align}
while $[\t]_\sim$ is the unique arrow such that for all $j=1,\dots,n$,
\begin{align}
  [\t]_\sim;[\pi_j]_\sim&= [\t;\pi_j]_\sim \tag{Functoriality of $Q_\eqsyn$}\\
  &= [p_j \cdot \t_j]_\sim \tag{\ref{eq:boh}}\\
  &=p_j\cdot [\t_j]_\sim. \tag{pca-enrichment of $\CatT{\Diag{\ThPB}}_\sim$}
\end{align}
Hence, since for all $j=1,\dots,n$, $[q_j \cdot \s_j]_\sim = [p_j \cdot \t_j]_\sim$, we have $[\s]_\sim = [\t]_\sim$.


For arrows of arbitrary type $\bigoplus_{j=1}^o A^{k_j} \to \bigoplus_{i=1}^nA^{m_i}$, one can easily rely on the universal property of coproducts and the case that we just proved.
\end{proof}

\begin{proof}[Proof of Corollary~\ref{cor:completenessPBPtapes}]
By Theorem~\ref{thm:completenessPBPtapes}, $I$ is faithful. The following derivation concludes the proof.
  \begin{align}
    \dsem{\s}=\dsem{\t} &\Longleftrightarrow (\CatT{Q_{\ThPB}}; Q_\eqsyn ; I) (\s)= (\CatT{Q_{\ThPB}}; Q_\eqsyn ; I) (\t) \tag{$\dsem{-}= \CatT{Q_{\ThPB}}; Q_\eqsyn ; I$}\\
    &\implies   (\CatT{Q_{\ThPB}}; Q_\eqsyn ) (\s)= (\CatT{Q_{\ThPB}}; Q_\eqsyn )(\t) \tag{$I$ is faithful} \\
    &\Longleftrightarrow \CatT{Q_{\ThPB}}  (\s) \eqsyn \CatT{Q_{\ThPB}}(\t) \tag{Def. $Q_\eqsyn$}\\
    &\Longleftrightarrow \s \dot{\sim} \t \tag{Proposition \ref{prop:twoequivalence}}
  \end{align}
\end{proof}

\begin{proof}[Proof of Corollary \ref{cor:finale}]
  \begin{align}
    \osem{c}=\osem{d} &\Longleftrightarrow  \dsem{\encoding{c}}=\dsem{\encoding{d}} \tag{Proposition \ref{prop:encoding}}\\
    &\Longleftrightarrow \encoding{c} \eqsynbis \encoding{d}. \tag{Proposition \ref{prop:soundnesstapes}, Corollary \ref{cor:completenessPBPtapes}}
  \end{align}
\end{proof}

    \end{document}